\documentclass[10pt,twoside]{article} 
\usepackage{times} 
\usepackage{xcolor}
\definecolor{darkblue}{rgb}{0,0,0.4}
\definecolor{darkgreen}{rgb}{0,0.5,0}
\definecolor{darkred}{rgb}{0.5,0,0}
\usepackage{hyperref}
\hypersetup{colorlinks,
linkcolor=darkred,
citecolor=darkgreen,
urlcolor=darkblue}

\usepackage{amsmath}
\usepackage{amsthm}
\usepackage{amsfonts}
\usepackage[all]{xy}
\usepackage{epsfig}   
\usepackage{verbatim}
\usepackage{eurosym}
\usepackage{color} 
\usepackage{comment}
\usepackage{multirow}
\usepackage{bbm} 
\usepackage{amssymb} 
\usepackage{braket} 
\usepackage{authblk} 

\def\by{\begin{color}{yellow}}
\def\ey{\end{color}}
\def\bg{\begin{color}{green}}
\def\eg{\end{color}}

\pdfoutput=1

\def\bR{\begin{color}{red}}  
\def\bB{\begin{color}{blue}}
\def\bM{\begin{color}{magenta}} 
\def\bC{\begin{color}{cyan}}
\def\bW{\begin{color}{white}} 
\def\bBl{\begin{color}{black}} 
\def\bG{\begin{color}{green}} 
\def\bY{\begin{color}{yellow}} 
\def\e{\end{color}}

\theoremstyle{definition}
\newtheorem{remark}{Remark}[section]
\newtheorem{rem}[remark]{Remark}
\newtheorem{example}[remark]{Example}

\theoremstyle{plain}
\newtheorem{thm}[remark]{Theorem}
\newtheorem{defn}[remark]{Definition}
	\newtheorem{lem}[remark]{Lemma}
	
	\newtheorem{prop}[remark]{Proposition}

\newcommand{\bit}{\begin{itemize}}
\newcommand{\eit}{\end{itemize}\par\noindent}
\newcommand{\ben}{\begin{enumerate}}
\newcommand{\een}{\end{enumerate}\par\noindent}
\newcommand{\beq}{\begin{equation}}
\newcommand{\eeq}{\end{equation}\par\noindent}
\newcommand{\beqa}{\begin{eqnarray*}}
\newcommand{\eeqa}{\end{eqnarray*}\par\noindent}
\newcommand{\beqn}{\begin{eqnarray}}
\newcommand{\eeqn}{\end{eqnarray}\par\noindent}

\def\C{{\bf C}}
\def\Cf{{\bf C}_{\rm free}}

\def\D{{\bf D}}
\def\St{{\rm S}}
\def\RC{{\rm PC}}
\def\UC{{\rm UC}}
\def\M{{\rm M}}
\def\FinProb{{\bf FinProb}}
\def\FinStoch{{\bf FinStoch}}
\def\Set{{\bf Set}}

\def\H{\mathcal{H}}
\def\blank{\underline{\phantom{x}}}

\newcommand{\lra}{\longrightarrow}

\title{A mathematical theory of resources}
\author[1]{Bob Coecke\thanks{Bob.Coecke@cs.ox.ac.uk}}
\author[2]{Tobias Fritz\thanks{tfritz@perimeterinstitute.ca}}
\author[2]{Robert W.~Spekkens\thanks{rspekkens@perimeterinstitute.ca}}
\affil[1]{University of Oxford, Department of Computer Science}
\affil[2]{Perimeter Institute for Theoretical Physics}

\begin{document}
\maketitle
\thispagestyle{empty}

\begin{abstract}
In many different fields of science, it is useful to characterize physical states and processes as resources.  Chemistry, thermodynamics, Shannon's theory of communication channels, and the theory of quantum entanglement are prominent examples. Questions addressed by a theory of resources include: Which resources can be converted into which other ones?  What is the rate at which arbitrarily many copies of one resource can be converted into arbitrarily many copies of another?  Can a catalyst help in making an impossible transformation possible? How does one quantify the resource? Here, we propose a general mathematical definition of what constitutes a resource theory.  We prove some general theorems about how resource theories can be constructed from theories of processes wherein there is a special class of processes that are implementable at no cost and which define the means by which the costly states and processes can be interconverted one to another. We outline how various existing resource theories fit into our framework. Our abstract characterization of resource theories is a first step in a larger project of identifying universal features and principles of resource theories.   In this vein, we identify a few general results concerning resource convertibility.
\end{abstract}

\tableofcontents

\section{Introduction}

In science, one can distinguish two traditions of theory-building.  On the one hand, there are the theories that seek to model, explain, and predict the natural behaviour of systems, in the absence of any human intervention and regardless of what anyone knows about them.   In physics, one could call this the \emph{dynamicist} tradition of theory building, where it is the most common approach.  On the other hand, there is the \emph{pragmatic} tradition wherein the goal is to describe the manner and extent to which a given system can be known and controlled by human agents. The more a scientific discipline is concerned with aspects close to human life and society, the more relevant this aspect is. The guiding philosophy in the pragmatic tradition is that understanding a phenomenon means being able to make use of it.   Physical phenomena are studied in order to better leverage certain \em resources.\em

Chemistry is a good example of a field that takes this pragmatic perspective.  Much of chemistry is about understanding chemicals as resources. 
This perspective can be traced back to the origin of the field, which was the study of alchemy, the goal of which was to find ways of transforming base metals such as iron, nickel, lead and zinc into noble metals such as silver and gold.  The pragmatic approach is still going strong today, particularly in modern industrial chemistry, which seeks to discover the processes by which raw materials that are available in abundance can be transformed into useful products, for instance, bulk chemicals such as gases, acids, bases, and petroleum, and secondary products such as dyes, pesticides, drugs, and polymers.   

A similar story is true of thermodynamics.  Early work, such as that of Carnot, sought to understand resources of thermal nonequilibrium in terms of their ability to do useful work, such as inducing mechanical nonequilibrium, and to determine the relative usefulness of different sorts of resources and different processes for extracting work from these.  This perspective still survives in modern treatments of the subject.

The basic questions which are answered in a theory of resources are: Which resources can be converted into which other ones? What are the ways in which a given conversion can be accomplished?

The first known rules of chemical transformations were determined empirically, as were the first constraints on thermodynamic transformations. Many resource theories start life in this fashion.  Eventually, however, scientists seek to provide a derivation of the rules of a resource theory from a deeper theory of the physics.  The rules of stoichiometry, for instance, can be inferred from Dalton's atomic theory, and the rules of thermodynamics can be inferred from statistical mechanics.  Typically such attempts at reconstruction generally lead to a better characterization of the resources and the possibilities of conversion.   For instance, after the development of the atomic theory of matter, the resources in chemistry were understood to be collections of molecules and the transformations among these were understood to be constrained by conservation of the constituent atomic elements.  

The resulting theories of resources incorporated many intuitive facts.  For instance, it is obvious that if one can convert resource $a$ to resource $b$ and one can also convert resource $b$ to resource $c$, then one can convert resource $a$ to resource $c$.  Another intuitive fact is that one can compose resources in parallel, and processes for implementing conversions among resources can be composed both in parallel and sequentially.  

The first goal of this article is to formalize these intuitive notions and develop a mathematical framework that captures precisely the assumptions of such resource theories.  We propose that a resource theory be associated with a symmetric monoidal category wherein the objects are resources and the morphisms are transformations of the resources.

This framework can then be studied as a mathematical entity in its own right.  One can also derive results in the abstract setting that can be fed back to all of the concrete resource theories.  Certain new questions can also be posed within the framework, e.g.~are there concepts, features and  general principles  that are common to all resource theories? 

As is often the case, by abstraction we arrive at a framework that can accommodate much more than the paradigmatic examples that led one to it.  

In the theory of computation, for instance, one is often interested to know whether a given set of gates, when combined in parallel and sequentially, can be used to build a circuit that allows the solution of computational problems in a given complexity class.  One can therefore think of gate sets as resources for achieving computational tasks.  As another example, in the theory of communication, one is interested to know whether certain kinds of noisy channels can be combined, together with pre- and post-processing, to simulate a noiseless channel, so that channels can be considered as resources for communication tasks.   Gates and channels, unlike chemicals and sources of heat, are transformations rather than states of matter.  Nonetheless, our abstract framework for resource theories can accommodate transformations as resources just as easily as it accomodates states.   Indeed, a given resource theory can include both resource transformations and resource states, as we shall see. 

Another sense in which the abstraction can accommodate a greater variety of examples is that the resources need not be considered to be \emph{physical} states or processes at all, they might instead be merely \emph{logical} entitites.  For instance, one can cast mathematics as a resource theory where the objects are mathematical propositions and the morphisms are proofs, understood as sequences of inference rules.   The tensor product is interpreted as logical conjunction, so that 
if $P$ and $Q$ are propositions, then $P\otimes Q$ corresponds to the proposition ``$P$ and $Q$''. Variants of this resource theory are studied in categorical logic and categorical proof theory; see e.g.~\cite{BHRU}.  The categorical formalism for resource theories that we propose will make a connection with this work. 

There is another schema by which many resource theories are defined.  For a given set of possible experimental interventions---for instance, state preparations, transformations and measurements---the set can be divided into those interventions that are considered to be freely implementable and those that are considered to be costly.    It is presumed that an agent can make use of anything in the free set in unlimited number and in any combination, while the elements of the expensive set are the resources. The theory seeks to describe the structure that is induced on the resources, given access to the free set. 

The best modern example of a resource theory arising in this fashion is entanglement theory.  The term `entanglement' was coined by Schr\"{o}dinger in the thirties, but the turning point in the theory of entanglement was in the mid-nineties, when researchers in quantum information realized that an entangled state was a resource.  One of the first examples of entanglement serving as a resource  was provided by the quantum teleportation protocol.  Suppose two agents, Alice and Bob, are restricted in the quantum operations that they can jointly perform.  Specifically, they cannot transmit quantum coherence between their labs, which is another way of saying that there is no quantum channel between them.  Nonetheless, it is assumed that there is a classical channel between them, and that \em within \em each of their labs, they have no trouble implementing any coherent quantum operation. This restricted set of operations is called the set of local operations and classical communication (LOCC). In the quantum teleportation protocol, a maximally entangled state is consumed, using LOCC, to simulate one use of a quantum channel~\cite{EntanglementResource}.  After the quantum information community had made the shift to thinking about entanglement as a resource, an entangled state was thereafter \em defined \em as a state that cannot be generated by LOCC.  

Many other properties of quantum states have been studied by quantum information theorists, with entanglement theory as the model. There are now resource theories of: 
\em asymmetry\em, where the resources are quantum states that break a symmetry and quantum operations that fail to be covariant with respect to a symmetry~\cite{gour2008resource,marvian2013theory,marvian2011pure,MarvianSpekkensNoether,marvian2012information,marvian2013modes};
\em nonuniformity\em, where the resources are quantum states that deviate from the maximally mixed state and quantum operations that deviate from those that are deterministic or inject uniform noise into the system \cite{PurityResource,nonequilibrium};  and \em athermality\em, where the resources are quantum states that deviate from the Gibbs form for a given temperature and quantum operations that deviate from those that are deterministic or inject thermal noise at a given temperature \cite{brandao2011resource}. Each of these will be discussed in this article.  The resource theory of asymmetry provides a framework for categorizing many constraints that arise from symmetries, such as selection rules in atomic physics, Noether's theorem~\cite{MarvianSpekkensNoether}, and the Wigner-Araki-Yanase theorem~\cite{ahmadi2013wigner,marvian2012information}.  The resource theories of athermality and nonuniformity provide a framework for describing many results in quantum thermodynamics, including the limits to work extraction~\cite{JWZGB,brandao2011resource,FundLimitsNature,SSP2}, Landauer's solution to the Maxwell demon problem~\cite{nonequilibrium}, and the status of the second law of thermodynamics~\cite{nonequilibrium,1ShotAtherm2}.  Note that there are strong parallels between the resource theories that arise in this way, and this provides yet another motivation for developing a general framework in which one can hope to distinguish features that are generic to all such resource theories and features that are particular to a given example.

This schema, once formalized, is quite generic and it too can be extended beyond the case of \em experimental \em processes to more abstract sorts of processes.
For instance, the resource theory of \em compass-and-straightedge constructions \em provides an example of this sort of phenomenon.  It has plane figures as its objects and compass-and-straightedge constructions between those as transformations. More precisely, a plane figure consists of a set of points, lines and circles in $\mathbb{R}^2$, and a transformation is a sequence of basic constructions, such as creating a new line through two existing points, creating a circle with a given center point and point on the radius, or dropping some of the existing points, lines, or circles.  In this context, certain plane figures become resources.  For example, a plane figure consisting of a circle together with a square of the same area is a valuable resource (Example~\ref{ex:catalyst}).  It cannot be created from a blank figure using compass-and-straightedge constructions, and with it, one can create many other figures. The question of interest is: which plane figures can be used to construct which others, given compass-and-straightedge constructions for free?

All such examples constitute theories that describe a set of processes and a distinguished subset thereof that are considered free.  The second goal of our article is to show how such theories---we will call them \em partitioned process theories\em---define resource theories in the sense of the abstract definition that we provide in the first part of the article.

Starting from a partitioned process theory, one can define resource theories of states, but also resource theories of generic processes, including states, transformations, and measurements.  Because transformations and measurements can admit of an input, one can combine them not only in parallel, but sequentially as well.  We shall discuss the difference between theories that permit only parallel-combinable resource processes, and those that allow them to be combined in any fashion.  Such theories of process transformations can be understood as a generalization, to arbitrary process theories, of the {\em quantum combs} framework~\cite{comb}. 

Finally, we come back to the abstract framework of resource theories, and we define a mathematical structure that captures only partial information about the resource theory, but information which is particularly significant.  

Of interest to us is the partial information that is required to answer some of the most basic questions about a resource theory: When are two resources equivalent under the free operations, in the sense that one can convert one to the other and vice-versa?  What are the necessary and sufficient conditions on two resources such that  the first can be converted to the second (possibly irreversibly) under the free operations?  How can we find measures of the quality of a resource, that is, functions from the resources to the reals which are nonincreasing under the transformations allowed by the resource theory?  What is the rate at which arbitrarily many copies of one resource can be converted to arbitrarily many copies of another, or equivalently, what is the average number of copies of a resource of one kind that is required to produce one unit of a resource of another kind?  Can a given resource conversion be made possible by the presence of a catalyst, i.e. a resource which must be present but is not consumed in the process?

The question of which \em particular \em free operation is used to achieve a given resource conversion is typically considered to be of secondary interest, and it is this information that we will dispense with in defining our more streamlined mathematical formalism.
The third goal of our article, therefore, is to provide a mathematical framework which provides just enough structure to specify the answers to such questions, and therefore to capture the ``core'' of a resource theory.  We refer to this minimal framework as a {\em theory of resource convertibility}.
The requisite mathematical formalism turns out to be that of a {\em commutative (equivalently, Abelian or symmetric) preordered monoid}.  

First, consider the preorder.  This is the order over resources that captures whether one resource can be converted to another or not.  
Recall that the first of the questions above called us to characterize the equivalence classes of resources under the free operations, and the second called us to find the {\em partial order} over these equivalence classes that is induced by the free operations.  But this is nothing more than a call to find the preorder over the resources induced by the free relation of convertibility.  Once the preorder is identified, one can define measures of the resource, or ``resource monotones'', as any map from the resources to the reals which respects the preorder.  

As an example, in the resource theory of bipartite entanglement, the preorder over pure bipartite states is determined by the eigenvalue spectrum of the reduced density operator of each such state.
One state is higher than another in the preorder if and only if the eigenvalue spectrum of the reduced density operator of the second state \em majorizes \em the eigenvalue spectrum of the reduced density operator of the first state~\cite{Nielsen}.  A given equivalence class of states contains all states for which the reduced density operator has a given eigenvalue spectrum (wherein the order of the eigenvalues is not significant).   Any Schur-convex function of the eigenvalue spectrum respects the majorization preorder and consequently defines an entanglement monotone.  The Shannon entropy, which is a Schur-convex function, defines the entanglement monotone known as the ``entropy of entanglement''~\cite{EntanglementResource}.

It should be noted that although it might appear that the highest goal of a resource theory should be to find useful {\em measures} of a resource, in fact any such measure is a rather crude way of expressing facts about the preorder.  The preorder structure of a resource theory is more fundamental.  Indeed, if the equivalence classes defined by the preorder do not form a total order, then having the valuation over the states for any single measure is insufficient for deducing the preorder over the states.  In particular, for two states that are not ordered relative to one another, one measure might assign a higher value to the first, while another might assign a higher value to the second.
As an example, although there have been many proposals of measures of entanglement for pure bipartite states, the majorization preorder over these states is more fundamental than any of these measures.    This point has been previously emphasized elsewhere~\cite{nonequilibrium}. 

Next, consider the monoid structure.  This is simply the binary operation that specifies the manner in which two resources compose in parallel.  
We here emphasize that the monoid structure of a resource theory can vary independently of the preorder structure.
To prove this point, we provide an example of two resource theories that are associated with the same preorder, but have different notions of parallel composition and hence that differ in their monoid structure.  

The monoid structure is what is relevant for answering the last two questions on our list, i.e. determining asymptotic rates of conversion and for establishing whether nontrivial catalysis is possible in the resource theory.  Determining the asymptotic rates of conversion among states in particular is traditionally one of the first topics to be studied in any given \emph{quantum} resource theory, under the name of \emph{distillation protocols}, for instance, for entanglement~\cite{EntDist}, asymmetry~\cite{gour2008resource,marvian2011pure}, and athermality~\cite{JWZGB,brandao2011resource}.  More generally, one seeks to determine the asymptotic rates of conversion among processes (including states as special cases).  The resource inequalities of quantum Shannon theory~\cite{DevetakResource1,DevetakResource2} are an example of such work.

The fact that different sorts of resources can parallel-compose differently, and in particular that not every resource has an extensive character, has practical significance.  Recognition of the latter fact is what ultimately resolved a puzzle concerning an apparent conflict between two results: on the one hand, Ref.~\cite{horodecki2002laws} seemed to establish that the asymptotic rate of conversion in any quantum resource theory of states was given by the regularization of the relative entropy distance of a state to the set of free states; on the other hand, Ref.~\cite{gour2008resource} had determined the asymptotic rate of conversion among asymmetric states and it was not given by the regularization of the relative entropy.  The resolution of the puzzle, described in Ref.~\cite{gour2009measuring}, was that the result of Ref.~\cite{horodecki2002laws} applied only to extensive resources, that is, resources where the relative entropy distance to the free states of $N$ copies of a state scaled linearly with $N$, while for the resource of asymmetry, the relative entropy  distance to the free states of $N$ copies scaled as $\log N$.  

A theory of resource convertibility is a particularly useful arena to try and identify some generic facts about resource theories, facts which might be more obvious in an abstract setting than in any concrete resource theory. As a small first example, we have derived a sufficient condition for the absence of catalysis in a resource theory.  A more substantive example 
would be to extend the results of Ref.~\cite{horodecki2002laws} by determining a general expression for the asymptotic rate of conversion of resources when the latter are not extensive in character.

\color{black}

\paragraph*{Acknowledgements.} 

We would like to thank the QPL referees as well as Manuel B\"arenz and Hugo Nava Kopp for detailed comments on the manuscript. Pictures have been produced with the TikZit package. 

Research at Perimeter Institute is supported by the Government of Canada through Industry Canada and by the Province of Ontario through the Ministry of Economic Development and Innovation. The first two authors have been supported by the John Templeton Foundation.

\section{Resource theories}

\subsection{What is a resource theory?}

We now outline the basic elements of a resource theory.   First, there may be many different kinds of resources, which we denote by $A, B, \ldots$.  In addition, there are conversions between resources, which we denote $f:A\to B, g:C\to D, \ldots$, much as in a chemical reaction.  A transformation $f:A\to B$ is a method for turning resource $A$ into resource $B$.
 Since there may be many ways of achieving such a conversion, we need to distinguish these by an additional label such as $f$. 
Transformations can be \em composed sequentially\em.  In particular, a transformation $f$ that converts resource $A$ to resource $B$ can be composed sequentially with a transformation $g$ that converts resource $B$ to resource $C$.  (This is only possible if the output of $f$ is the same as the input of $g$, in which case we say that the types match.)  The composite transformation $f \circ g$ converts resource $A$ to resource $C$.

Moreover, it should be possible to combine resources: if $A$ and $B$ are resources, then it is possible to regard $A$ and $B$ together as a composite resource, which we denote $A\otimes B$. In other words, the collection of all resources should be equipped with a binary operation ``$\otimes$''. If $f:A\to B$ and $g:C\to D$ are transformations, then there should similarly be a composite transformation $f\otimes g$ corresponding to executing $f$ and $g$ in parallel,
\[
f\otimes g \: : \: A\otimes C \lra B\otimes D.
\]
Finally, we assume the existence of a \em void resource \em denoted $I$, which can be added to any other resource without changing it.\footnote{For many physicists and mathematicians, the ``$\otimes$'' notation brings to mind Hilbert spaces, but in this context, it means simply a parallel composition of resources.  For instance, if one is interested in cooking ingredients as resources, then the set of ingredients consisting of a potato and a carrot is denoted ${\rm potato} \otimes {\rm carrot}$~\cite{CatsII}.}

All of these intuitive facts about resources, if formalized, can be summed up by saying that resources and transformations among resources are, respectively, objects and morphisms in a symmetric monoidal category (SMC). We do not repeat here the axioms defining the structures in an SMC, but instead refer the reader to the literature.  For instance, Ref.~\cite{Cats,CatsII} provides accessible expositions for a physics audience, while Ref.~\cite{Benabou} is among the original works.  Moreover, SMCs admit an intuitive graphical calculus, and much of this paper can be understood in terms of it.  In fact, there exist powerful theorems which establish that equational reasoning within an SMC is in one-to-one correspondence with deformation of diagrams \cite{JS,SelingerSurvey}. 

In the graphical calculus of an SMC, we represent an object $A$ by a wire labeled with that object:   
\[
\includegraphics{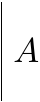}
\]
a composite object $A_1\otimes \dots \otimes A_n$ by placing such wires side-by-side:
\[
\includegraphics{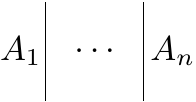}
\]
and a general morphism $f:A_1\otimes \dots \otimes A_n\to B_1\otimes \dots \otimes B_m$ by a box:
\[
\includegraphics{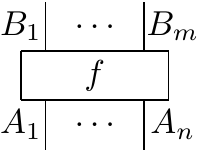}
\]
The trivial object is represented by `no wire', so a morphism $s: I \to A$ has no input wires.  Rather than as a box, we represent it as a triangle:
\[
\includegraphics{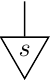}
\]
We typically omit the object labels whenever no danger of confusion arises. The cautious reader may note a similarity between the triangle depicting a state and Dirac notation in quantum theory, and indeed, this graphical calculus can be seen as Dirac notation, unfolded in two dimensions \cite{Cats, CatsII}. `Symmetry' in symmetric monoidal category stands for the fact that wires may cross:
\[
\includegraphics{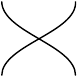}
\]

Sequential composition of $f:A\to B$ and $g:B\to C$ is represented by connecting $f$'s output to $g$'s input:
\[
\includegraphics{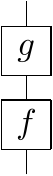}
\]
and parallel composition of $f_1:A_1\to B_1$ and $f_2:A_2\to B_2$ by placing boxes side-by-side:
\[
\includegraphics{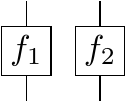}
\]
This completes the notation of the graphical calculus.  Equational reasoning boils down to nothing but deforming diagrams without changing the topology:
\[
\includegraphics{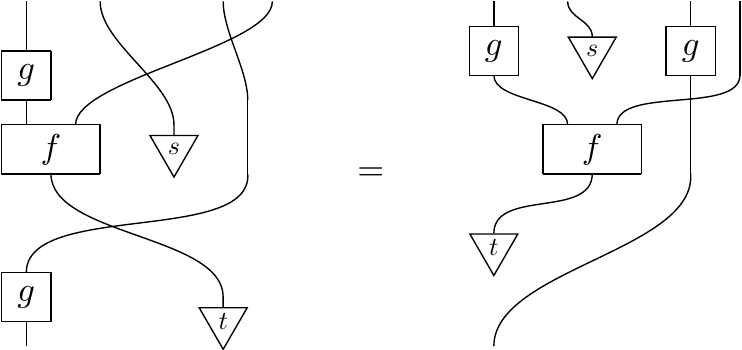}
\]
As already mentioned above, doing so allows one to establish any equation that can be derived from the axioms of an SMC.

We adopt the convention that if $\C$ is an SMC, then $|\C|$ denotes the set of its objects, and $\C(A,B)$ denotes the hom-set between $A$ and $B$, that is, the set of morphisms in $\C$ that have $A$ as domain and $B$ as codomain. 

To summarize then:

\begin{defn}
\label{def:resourcetheory}
A \emph{resource theory} is a symmetric monoidal category $(\D,\circ,\otimes,\mathbb{I})$, where:
\begin{itemize}
\item the objects $|\D|$ represent the resources, 
\item the morphisms in the hom-set $\D(A,B)$ for $A,B\in |\D|$ represent transformations of resource $A$ into resource $B$ that can be implemented without any cost,
\item the binary operation $\circ$ corresponds to sequential composition of processes with matching types,
\item the binary operation $\otimes$ describes parallel composition of resources and of processes, and 
\item the unit object $\mathbb{I}$ denotes the void resource, i.e., the resource which when composed with any other resource leaves it the same, as in $A\otimes \mathbb{I}=A=\mathbb{I}\otimes A$.
\end{itemize}
\end{defn}

The difference between the terms ``resource theory'' and ``symmetric monoidal category'' is not a mathematical one, but rather one of interpretational nature: a particular SMC is called a ``resource theory'' whenever we want to think of its objects as resources and its morphisms as transformations or conversions between those resources. So a resource theory, as an abstract entity, is just a symmetric monoidal category.  A concrete resource theory, such as chemistry or thermodynamics, is a particular interpretation of that symmetric monoidal category.  

Note that given the interpretation of the unit object $\mathbb{I}$ as the void resource, it necessarily has no cost.  Given that the morphisms of the resource theory are interpreted as transformations that can be implemented at no cost, it follows that any object $A$ that can be obtained from the void resource, i.e. for which the hom-set $\D(\mathbb{I},A)$ is non-empty, also has no cost.  The set of all such objects will be called the {\em free resources}.  The complement of this set will be the {\em costly} or {\em nonfree resources}.  
\begin{defn}
The set of {\em free resources} in $\D$ is $\{ A \in |\D | : \D(\mathbb{I},A) \ne \emptyset \}$.
\end{defn}

\begin{remark}
\label{rem:strict}
Throughout this paper, we assume that $(\D,\circ,\otimes,\mathbb{I})$ is a symmetric `strict' monoidal category, that is the unit and associativity natural isomorphisms are identities. By Mac Lane's strictification theorem~\cite[p.257]{MacLane}, any (non-strict) monoidal category is  equivalent (via a pair of strong monoidal functors) to a strict one. In particular, to show that a certain equation holds for general monoidal categories, it suffices to show this in the strict case.  Hence proofs in the graphical calculus carry over to arbitrary monoidal categories.  Concretely, in the graphical calculus, the associator becomes trivial since bracketing vanishes, and unitality becomes trivial since the tensor unit is represented by `no wire'. 

On the other hand, the symmetry isomorphisms $A\otimes B\stackrel{\cong}{\to} B\otimes A$ cannot be taken to be identities, and also in the graphical calculus, they need to be represented by the crossing of strands. One can either think of them as actual transformations which convert $A\otimes B$ into $B\otimes A$, or as ``passive'' transformations merely expressing the fact that $A\otimes B$ and $B\otimes A$ stand for the same objects, i.e.~the order in a tensor expression has no significance beyond syntax. 

Consult \cite{CatsII} for a discussion of all these issues.
\end{remark}

\subsection{Examples}

We now give a few examples of resource theories and how their structure can be formalized using a SMC.   We begin with our chemistry example.

\begin{example}
\label{ex:chem}
The resource theory of \em chemistry \em has collections of chemical species (atoms, ions, molecules) as its objects, and reactions under standard conditions as well as sequences of such reactions as transformations. An example of a simple reaction is the neutralization of an acid and a base,
\[
\mathrm{NaOH + HCl} \lra \mathrm{NaCl + H_2 O}.
\]
Here, we have used the usual notation in chemistry, where the ``$+$'' corresponds to our tensor ``$\otimes$''. The tensor unit $\mathbb{I}$ is simply the empty set of chemical species.  
An example from industrial chemistry of a transformation that requires a sequence of reactions is the widely used \em Haber process, \em 
\[
\mathrm{N_2 + 3 H_2} \lra \mathrm{2 NH_3},
\]
turning nitrogen and oxygen into ammonia, which can then be further processed into fertilizer.   It should be noted that different sequences of reactions relating the two sides of a reaction equation correspond to different morphisms in the category.

\end{example}

\begin{example}
\label{ex:rand}
The resource theory of randomness is defined as follows.
An object $(X,p)$ is a finite set equipped with a probability distribution $p$ which assigns to every $x\in X$ its probability $p(x)$. The transformations $f:(X,p)\to (Y,q)$ correspond to deterministic processings, which are deterministic maps $f:X\to Y$ having the property that $q(y) = \sum_{x\in f^{-1}(y)}p(x)$, meaning that one can compute the probability of getting a certain $y\in Y$ after the processing by summing up the individual probability of all $x\in X$ which process to $y$.  These maps compose sequentially in the obvious way. This defines a category $\FinProb$, which has previously been studied in the context of entropy~\cite{BFL}. To get from $\FinProb$ to an SMC that describes the resource theory of randomness, we must formalize the notions of parallel composition and of the trivial resource.
Parallel composition of objects is given by taking product distributions:
\[
(X,p) \otimes (Y,q) := (X\times Y, p\times q),
\]
where  $X\times Y$ is the Cartesian product of the sets $X$ and $Y$, and $p\times q$ denotes the product distribution on $X\times Y$. Parallel composition of deterministic processings is defined in the obvious way: when we have $f:(X,p)\to (X',p')$ and $g:(Y,q)\to (Y',q')$, then it is straightforward to check that
\[
f \times g \: : \: (X\times Y, p\times q) \to (X'\times Y', p'\times q') ,\qquad (x,y) \mapsto (f(x), g(y))
\]
is also a deterministic processing, and this is how we define the parallel composition $f\otimes g : (X,p)\otimes (Y,q) \to (X',p')\otimes (Y',q')$. 

Finally, the trivial resource $\mathbb{I}$ is taken to be the singleton set equipped with the point distribution.
The SMC that is obtained by augmenting $\FinProb$ with this tensor and tensor unit
$(\FinProb,\circ,\otimes,\mathbb{I})$
 defines the resource theory of randomness.

This theory has many practical applications.  For example, randomness is a valuable resource for secure cryptography. 
In many practical schemes for IT security, a computer gathers ``entropy'' from user input like keyboard strokes or mouse movement. Having ``bad'' randomness, in the sense of predictability, may lead to critical security issues. For this reason, true randomness is an important resource, and our definitions allow for a precise mathematical treatment of the question of the quality of randomness.  Our formalism also allows for a treatment of the manipulation of randomness. For instance, in the theory of  \em randomness extractors \em  one wants to find a deterministic transformation $f$ which turns a given $(X,p)$ with non-uniform $p$ into a $(Y,q)$ with uniform distribution $q$ (for this to be possible, the cardinality of $Y$ must be less than that of $X$).  Finding randomness extractors can therefore be understood as a problem in the resource theory of randomness. 

\end{example}

Lest the reader get the impression that the objects in a resource theory are always {\em states}, we present an example wherein the resources are themselves processes which can be interconverted one to the other by pre- and post-processing.  (Resource theories containing processes among the resources will be an important theme in subsequent sections.)

\begin{example}
\label{ex:shannon}
The resource theory of \em one-way classical communication channels \em concerns the mathematical theory of communication as developed by Shannon~\cite{Shannon}. In the associated SMC, an object is a triple $(A,B,P)$, where $A$ and $B$ are finite sets, and $P$ is a conditional probability distribution over the second set given the first, with $P(b|a)$ denoting the probability of $b$ given $a$ for all $a\in A$ and $b\in B$.
We think of the ingoing value $a$ as a message that Alice wants to send to Bob, who receives $b$. The transmission of the message may not be perfect, and hence Bob's $b$ may differ from Alice's $a$; more precisely, for every value of $a$ there is a certain probability of getting each possible value of $b$, and this is described by the given conditional distribution $P(b|a)$. Parallel composition of objects $(A,B,P)$ and $(C,D,P')$ corresponds to forming products of stochastic maps,
\[
(P\otimes P')(bd|ac) := P(b|a) P'(d|c),
\]
for all $a\in A$, $b\in B$, $c\in C$, and $d\in D$.
If Alice and Bob have access to the communication channel $P(b|a)$, they can use this channel to simulate certain other channels $Q(b'|a')$ by Alice applying a stochastic map $E(a|a')$ to the input of the channel, termed an \em encoding \em of $a'$, and Bob applying a stochastic map  $D(b'|b)$ to its output, termed a \em decoding \em of $b'$.  In this way, the channel resource is transformed as follows:
\[
P(b|a) \to Q(b'|a')
\]
where
\[
Q(b'|a'):= \sum_{a,b} D(b'|b) P(b|a) E(a|a'). 
\]

We also allow shared randomness as a free resource, that is, Alice and Bob are allowed to possess a pair of systems in an arbitrary correlated state.
Using this free resource, Alice and Bob can correlate the encoding and decoding maps and thereby transform the channel as follows:
\[
P(b|a) \to Q(b'|a'):= \sum_{a,b,x,y} D(b'|b,y) P(b|a) E(a|a',x) R(x,y), 
\]
where $x$ and $y$ denote the correlated variables in Alice's and Bob's possessions respectively, and $R(x,y)$ is the joint distribution describing the correlation. 

These are the morphisms in the SMC defining the resource theory, which compose sequentially and in parallel in the obvious way. Finally, the unit object  $\mathbb{I}$ in the SMC is the deterministic map from the singleton set to the singleton set, representing the channel which ``does nothing to nothing''.

The principal goal of Shannon's information theory is to find an encoding and a decoding such that the composite channel $Q(b'|a')$ is as close as possible to an identity channel $Q(b'|a') = \delta_{b',a'}$ (for that to be possible, the cardinality of the respective domain of $a'$ and of $b'$ must again be generally less than that of $a$ and of $b$). In this way, Alice and Bob can simulate a noiseless transmission of information.

\end{example}

\section{Resource theories from partitioned process theories}\label{sec:ProcTheor} 

\subsection{Partitioned process theories}

In recent years, SMCs have proven to be a 
convenient mathematical framework for abstractly describing theories of physical processes~\cite{Cats,CatsII} as well as more general kinds of processes~\cite{BaezLNP, Gospel}.
Above we interpreted the objects of an
SMC as resources and the morphisms as transformations thereof.  \em
Process theories \em are SMCs as well, although the interpretation is now
different.  The objects in the SMC, $|\C|$, now correspond to the
different types of systems in the process theory.  The morphisms
correspond to the different processes (including state preparations
and measurements in the case of physical processes); in particular,
the processes that have a system $A$ as input and system $B$ as
output correspond to the elements of the hom-set $\C(A,B)$.
 The notion of executing physical processes in series is represented by composition of the morphisms, $\circ$, while the notion of parallel composition of systems and physical processes is given by the tensor product $\otimes$ in the SMC.  Finally, the trivial system, i.e. ``nothing'', corresponds to the unit object $I$ in the SMC.  A couple of examples serve to illustrate the concept of a theory of physical processes.

\begin{example}\label{ex:classicalstochasticprocesses}
The theory of classical stochastic processes for systems with discrete state spaces is defined as the SMC  in which the objects are finite sets,
$A,B,\dots$, and the morphisms are stochastic maps between these, that is, the hom-set $\C(A,B)$ is the set of conditional probabilities $P(b|a)$ for all $a\in A$ and $b\in A$, i.e.~arrays of numbers $P(b|a)$ indexed by $a\in A$ and $b\in B$ such that $P(b|a) \ge 0$ and $\sum_b P(b|a)=1$ for all $a$. Sequential composition of stochastic maps is essentially matrix multiplication,
\[
(Q\circ P)(c|a) = \sum_b Q(c|b) P(b|a),
\]
and in this way one obtains a category which is sometimes called $\FinStoch$~\cite{BF}. To turn this category into an SMC, we augment it by a notion of parallel composition and a unit object~\cite{QClChannels}.
Parallel composition of objects is given by the Cartesian product of the finite sets, i.e. $A \otimes B := A \times B$, and parallel composition of stochastic maps is given by their entrywise product:
\[
(P \otimes Q)(bb'|aa') = P(b|a) Q(b'|a').
\]
The unit object of the process theory is the singleton set denoted $I$.  Therefore the SMC for the theory of classical stochastic processes is $(\FinStoch,\circ,\otimes,I)$.

While {\bf FinStoch} and {\bf FinProb} (from example~\ref{ex:rand}) are obviously related,
there are some clear differences.  While in {\bf FinProb}
probability distributions are encoded as objects, in {\bf FinStoch}
they are encoded as the morphisms.  This clearly illustrates the
difference between an SMC representing a resource theory and one
representing a process theory.
\end{example}

\begin{example}\label{ex:quantumprocesses}
The theory of quantum processes on finite-dimensional system is defined as the SMC in which objects are finite-dimensional Hilbert spaces, morphisms are completely positive maps which compose sequentially in the obvious way, parallel composition is given by the tensor product of Hilbert spaces, and the unit object is the 1-dimensional Hilbert space $\mathbb{C}$.
 \end{example}

 
In the examples that are going to follow, the overarching SMC $(\C,\circ,\otimes,I)$ defining the process theory is going to be a variant of either the SMC of classical stochastic processes or the SMC of quantum processes, in the sense of being equivalent to either of these two as a symmetric monoidal category, that is, these variants simply introduce some additional structure, which when forgotten, leaves us with the original process theory.

\begin{defn}
\label{def:processtheory}
A \em partitioned process theory \em 
consists of a process theory, described by an SMC $(\C,\circ,\otimes,I)$, and a distinguished subtheory of \em free processes\em, described by an SMC $\Cf$ that is an all-object-including sub-SMC of $\C$: 
\[
\Cf \hookrightarrow \C\ . 
\]
We will denote a given partitioned process theory by $(\C,\Cf)$.
\end{defn}

Some explanation is in order. The ``free processes'' forming $\Cf$ are assumed to be those which our agent can execute at no cost.
 If $f$ and $g$ are free processes, then clearly $f\otimes g$ should also be a free process, since $f$ and $g$ can in particular be executed in parallel. Likewise, if $f\circ g$ is defined, then it should also be a free process. This justifies the assumption that $\Cf$ is a sub-SMC of $\C$. Since doing nothing to a system should certainly be a free process, it follows that the identity map on any system should be included in $\Cf$. In other words, $\Cf$ should be a subcategory of $\C$ that includes all objects of $\C$. In concrete applications, it is frequently the case that $\Cf$ is not directly given as an all-object-including sub-SMC, but in terms of a ``generating set'' of processes which one declares to be free, and then one has to close this set under parallel and sequential composition in order to obtain the smallest all-object-including sub-SMC that contains it, and this is then what $\Cf$ ends up being in such a case.  The set of processes that are \em not \em in the free set, that is, $\C\setminus \Cf$, are the costly resources.  Obviously, in order to obtain an interesting resource theory, it should be the case that $\Cf\neq\C$. 

The requirement $\Cf\neq\C$ is a constraint that may not be straightforward to verify given a specification of $\Cf$ in terms of a generating set of free operations.  For instance, in the context of the theory of quantum processes, consider the following question: if we allow all unitaries to be included in $\Cf$, as well as all partial trace operations, then what states can we include in addition and still obtain a nontrivial resource theory?  It turns out that if we allow for processes to be simulated only approximately (see comments about epsilonification in Sec.~\ref{closing}), then there is only a single type of state that one can add without making the closure of the free processes into the full set of processes, namely, completely mixed states~\cite{PurityResource,nonequilibrium}. One thereby obtains the resource theory of nonuniformity, Example~\ref{examplenonuniformity}.  If we add \emph{any other state} to the set of free processes, then after closing the set under sequential and parallel composition, we make it possible to implement an approximation of every quantum process for free, i.e. we obtain the trivial theory wherein $\Cf=\C$.

Another point to note is that although one may be able to mathematically define many distinguished subtheories of a given process theory, the result may not have any physical interest.  In order to be physically interesting, it should be the case that one can identify a set of constraints on experimental interventions which limit the processes that can be implemented in unlimited number to all and only those included in $\Cf$.

We turn now to showing how a partitioned process theory, as in Definition~\ref{def:processtheory}, can be proven to define many different resource theories, in the sense of Definition~\ref{def:resourcetheory}.  In particular, we get different resource theories by considering different types of processes.

\subsection{Resource theories of states}

We will begin by outlining how a partitioned process theory $(\C,\Cf)$ gives rise to a resource theory of states. By definition, a \emph{state} is a process whose input is the trivial object; in other words, ``state'' is a different term for ``preparation procedure''. For our partitioned process theory, the set of states is $\bigcup_{A\in|\C|}\!\!\C(I,A)$.  
 We obtain a resource theory by considering 
  the extent to which these states can be converted one to another by the free transformations, that is, 
the structure of $\bigcup_{A\in|\C|}\!\!\C(I,A)$ under $\bigcup_{A, B\in|\C|}\!\!\Cf(A, B)$.

More precisely, we define a resource theory of states in terms of a partitioned process theory $(\C,\Cf)$ as follows. The category representing the resource theory of states that we obtain from $(\C,\Cf)$ will be denoted $\St (\C,\Cf)$.  
The objects of $\St (\C,\Cf)$ are taken to be
 the states of $\C$,
\[
|\St (\C,\Cf)|:= \!\!\bigcup_{A\in|\C|}\!\!\C(I,A).
\] 
A state $s:I\to A$ can be converted into another state $t:I\to B$ by a free transformation $\xi:A\to B$ if one has
\beq
\begin{split}
\label{eq:stateconvert} 
\includegraphics{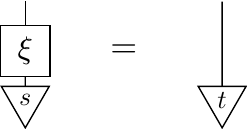}
\end{split}
\eeq
We then define the hom-set $\St (\C,\Cf)(s,t)$ for $s,t \in |\St (\C,\Cf)|$ to be the set of such free transformations that achieve $s\to t$.

We now define two transformations to be equivalent if they have the same effect on all states, including when acting on states of a composite system of which the transformation only acts on one part.  The morphisms in the resource theory of states are the equivalence classes of free transformations.  Note, that for all cases listed in Table 1, however, any two distinct free transformations are inequivalent.~\footnote{Nonetheless, it is useful to explicitly define this equivalence relation to make it clear that $\St(\C,\Cf)$ is  a subcategory of the resource theory of parallel-composable processes, $\RC(\C,\Cf)$, which we consider in the next subsection.}

So the hom-set $\St (\C,\Cf)(s,t)$ for $s,t \in |\St (\C,\Cf)|$ is the set of equivalence classes of such free processes that achieves $s\to t$, that is, 
\[
 \St (\C,\Cf)(s,t) := \{ \xi \in \Cf(A, B): \xi \circ s = t\} /\sim .
\]

Sequential composition in this category is defined as follows: if 
$\xi\in\Cf(A,B)$
 is a free process turning $s$ into $t$, and 
$\eta\in\Cf(B,C)$ 
is another free process turning $t$ into a third state $u:I\to C$, then $\eta\circ\xi \in \Cf(A,C)$ is a free process turning $s$ into $u$,
\[
\includegraphics{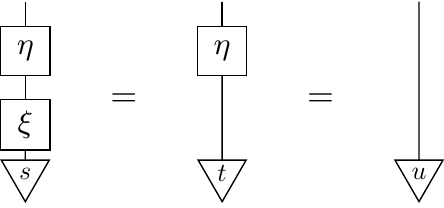}
\]
It is straightforward to check that this respects the equivalence, and moreover that the resulting composition on equivalence classes
turns 
$\St (\C,\Cf)$
 into a category. In order for it to be a symmetric monoidal category, we also need to define parallel composition of objects and morphisms. On objects in  $\St (\C,\Cf)$, which are states in $\C$, parallel composition is simply inherited from $\C$. Morphisms in 
 $\St (\C,\Cf)$,
which are equivalence classes of transformations between states in $\C$ using free processes, can also be composed in parallel in the obvious way,
\[
\includegraphics{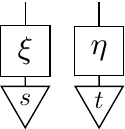}
\]
where we use the fact that $\Cf$ is closed under parallel composition. Compatibility with the equivalence relation is again easy to show. Since we assume $\C$ to be a strict monoidal category, it follows that 
$\St (\C,\Cf)$
 is a strict monoidal category as well. Finally, the symmetry isomorphism $s\otimes t \to t\otimes s$ on objects in $\St (\C,\Cf)$ is given by the symmetry in $\C$,
\[
\includegraphics{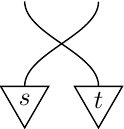}
\]
This symmetry must be a free process since $\Cf$ was assumed to be an all-object-including symmetric monoidal subcategory of $\C$, which entails in particular that it inherits the symmetries from $\C$. 

Finally, the identity object $\mathbb{I}$ for the SMC $\St (\C,\Cf)$ is the tensor unit of $\C$, $1_I$.

All told, we have proven the following:
\begin{thm}
\label{thm:statetorec}
For any partitioned process theory $(\C,\Cf)$, the procedure outlined above allows one to define a symmetric monoidal category $\St (\C,\Cf)$ that can be interpreted as a resource theory in the sense of Definition~\ref{def:resourcetheory}.
\end{thm}

Mathematically, 
$\St (\C,\Cf)$ 
is the \em coslice category \em of $\C$ over the unit object $I$, with the additional restriction that only processes in $\Cf$ are allowed for turning one state into another.

We begin by showing how the resource theory of randomness, Example~\ref{ex:rand}, arises from a partitioned process theory.
\begin{example}
\label{ex:rand2}
The process theory $(\C,\circ,\otimes, I)$ from which the resource theory of randomness arises is the theory of classical stochastic processes, defined in Example~\ref{ex:classicalstochasticprocesses}.

The distinguished subtheory of free processes $\Cf$ is the subcategory of classical deterministic processes.  These can be described as the subset of stochastic maps for which all the conditional probabilities are 0 or 1.  For any pair of objects $X$ and $Y$ in $\C$, we can define the deterministic maps from $X$ to $Y$.  Clearly, the property of being deterministic is preserved under parallel and sequential composition, and all identity maps are deterministic.  We have therefore confirmed that the deterministic maps form an all-objects-including subcategory of $\C$.

As usual in process theories, states on $X$ are simply a particular kind of stochastic map, the map from the singleton set to $X$.  It follows that the free states are simply the point distributions on $X$.  Every state that is not a point distribution, i.e. every distribution containing some randomness, is a nonfree resource.  

If we denote the SMC for the resource theory of randomness of Example~\ref{ex:rand} by $(\D,\circ,\otimes,\mathbb{I})$, then what we have shown is that for $\C$ the SMC of classical stochastic processes and $\Cf$ the sub-SMC of classical deterministic processes, the resource states $|\St (\C,\Cf)|$ can be identified with the objects of $\D$, and the transformations in the hom-set $\St (\C,\Cf)(s,t)$ for $s,t\in |\St (\C,\Cf)|$ are the transformations in the hom-set $D(s,t)$.  The unit object of the resource theory $\mathbb{I}$ is $1_I$, the identity map on the unit object of the process theory, which is the identity map applied to the singleton set.

\end{example}

\newcommand{\tstrut}{\rule{0pt}{2.6ex}}
\newcommand{\bstrut}{\rule[-1.1ex]{0pt}{0pt}}
\begin{table} 
\label{tab:processrecs}
\begin{center}
\begin{tabular}{|c|c|c|} 
 resource & systems & free processes  \\ \hline \hline
\multirow{ 2}{*}{bipartite entanglement~\cite{EntanglementResource}} & pairs of & \multirow{2}{*}{bipartite LOCC operations} \tstrut\\
 & Hilbert spaces & \bstrut\\\hline
\multirow{2}{*}{$n$-partite entanglement} & $n$-tuples of & \multirow{2}{*}{$n$-partite LOCC operations} \tstrut\\
 & Hilbert spaces & \bstrut\\ \hline
asymmetry~\cite{gour2008resource,marvian2013theory} & pairs of a Hilbert space & \multirow{2}{*}{$G$-covariant operations} \tstrut\\
 (relative to a symmetry group $G$) & \& a unitary rep'n of $G$ &  \bstrut\\ \hline
\multirow{2}{*}{nonuniformity~\cite{PurityResource,nonequilibrium}} & \multirow{2}{*}{Hilbert spaces} & \multirow{2}{*}{noisy operations} \tstrut\\
 & &  \bstrut\\ \hline
 athermality~\cite{JWZGB,brandao2011resource,FundLimitsNature} & pairs of a Hilbert space & 
 \multirow{2}{*}{$T$-thermal operations} \tstrut\\
(relative to temperature $T$) &  \& a Hamiltonian &  \bstrut\\ \hline
\end{tabular}
\end{center}
\caption{The resource theories of processes which have been most studied so far in quantum information theory. 
 In all cases, the 
 states and processes are all quantum states and quantum channels, i.e.~completely positive maps.}
\end{table}

In the following, we provide some examples of {\em quantum} resource theories of states that have been derived from partitioned process theories.
See Table~\ref{tab:processrecs} for a summary.

\begin{example}\label{examplenonuniformity}
The quantum resource theory of \em nonuniformity\em~\cite{PurityResource,nonequilibrium} is defined in terms of the following partitioned process theory.  The enveloping process theory is the SMC of quantum processes of Example \ref{ex:quantumprocesses}, and the free processes consist of the sub-SMC which is generated by three kinds of processes: first, preparing any system in the completely mixed state; second, applying any unitary transformation to a system; third, discarding any system by taking a partial trace. These free processes were called \em noisy operations \em in \cite{PurityResource}, (although this terminology seems a bit unfortunate since introducing noise is neither necessary nor sufficient for satisfying the required condition). Equivalently, we can characterize the free processes as those that have a Stinespring dilation for which the ancilla state is completely mixed. It follows from the definition of the free processes that a system with Hilbert space $\H$ has only a single free state, namely, the completely mixed state on $\H$.  Any state which is not completely mixed on $\H$---in other words, any state that is \em not uniformly mixed\em---is a nonfree resource.
\end{example}

\begin{example}
\label{ex:bientangle}
The most prominent example of a resource theory in the field of quantum information 
is the resource theory of \em bipartite entanglement\em. The enveloping process theory is a variant of the theory of quantum processes which contains some additional structure that is used to pick out the distinguished subtheory. The systems are pairs of finite-dimensional Hilbert spaces $(\H_A,\H_B)$ describing the systems owned by Alice and Bob, respectively.  The processes of type $(\H_A,\H_B)\lra(\H'_A,\H'_B)$ are quantum channels $\mathcal{L}(\H_A\otimes \H_B) \lra \mathcal{L}(\H'_A\otimes\H'_B)$, i.e.~completely positive trace-preserving maps turning an operator on $\H_A\otimes\H_B$ into an operator on $\H'_A\otimes\H'_B$. Sequential composition of such processes is the usual one. Parallel composition of these systems is component-wise,
\[
(\H_A,\H_B)\otimes (\H'_A,\H'_B) := (\H_A\otimes \H'_A,\H_B\otimes\H'_B),
\]
and similarly for the processes. The unit system is $I=(\mathbb{C},\mathbb{C})$. This describes the SMC $(\C,\circ,\otimes,I)$ of the process theory. As far as this category is concerned, the splitting into Alice's and Bob's Hilbert space is irrelevant: the definition of processes and their composition only involves the product $\H_A\otimes\H_B$ rather than $\H_A$ and $\H_B$ individually. 
The distinction between A and B is only relevant for defining the set of free processes, in the precise sense that the SMC $(\C,\circ,\otimes,I)$ is equivalent to the usual SMC of quantum processes in which the objects are Hilbert spaces and the morphisms are completely positive trace-preserving maps, both with the usual tensor product.

The free processes are defined to be the processes in the sub-SMC $\Cf$ corresponding to local operations and classical communication. Here, the splitting into a Hilbert space for Alice and a Hilbert space for Bob is crucial, since only this splitting enables one to know what the terms ``local'' and ``communication'' refer to. More precisely, a process $(\H_A,\H_B)\lra (\H'_A,\H'_B)$ is \em local \em if it is given by a completely positive map $\Phi : \mathcal{L}(\H_A\otimes \H_B) \lra \mathcal{L}(\H'_A\otimes\H'_B)$ which factors into a product
\[
\Phi_A \otimes \Phi_B \: : \:  \mathcal{L}(\H_A\otimes \H_B) \lra \mathcal{L}(\H'_A\otimes\H'_B),
\]
where $\Phi_A \: : \:  \mathcal{L}(\H_A) \lra \mathcal{L}(\H'_A)$ and $\Phi_B \: : \:  \mathcal{L}(\H_B) \lra \mathcal{L}(\H'_B)$.
To define classical communication, we make use of the subset of quantum processes that can be achieved by measuring some preferred basis of the Hilbert space and then repreparing a state that is diagonal in this basis as a function of the measurement outcome. A one-way classical communication is then a process of type $(\H_A,\mathbb{C})\to (\mathbb{C},\H_B)$ or of type $(\mathbb{C},\H_B)\to(\H_A,\mathbb{C})$, but drawn from this special subset of ``decohering'' processes. We can now define $\Cf$ to be the smallest sub-SMC of $(\C,\circ,\otimes,I)$ which contains all local operations and classical communication. As a particular case, one finds that the free processes of type $(\mathbb{C},\mathbb{C})\to (\H_A,\H_B)$, that is, the free states, are precisely the separable states on $\H_A\otimes\H_B$, which are those of the form
\[
\rho_{AB} = \sum_i w_i \rho^A_i \otimes \rho^B_i,
\]
where $(w_i)$ is a probability distribution and $\rho^A_i$ and $\rho^B_i$ are density operators on $\H_A$ and $\H_B$ respectively.  It follows that a state  is a nonfree resource only if it is {\em not} of the separable form.  These, therefore, are the bipartite entangled states.
\end{example}

\begin{example}
\label{ex:entangle}
The theory of \em $n$-partite entanglement \em is very similar: systems are $n$-tuples of Hilbert spaces, and processes are completely positive trace-preserving maps between their tensor products. The free processes are again the ones obtained from local processes and classical communication, wherein local in this context means factorization across the tensor product of $n$ Hilbert spaces, and classical communication among the $n$ parties is any combination of classical communication between any pair of parties.

The theory of bipartite entanglement can be regarded as a subtheory of $n$-partite entanglement by considering only those $n$-tuples of Hilbert spaces for which all components except for two are equal to $\mathbb{C}$. The theory of $k$-partite entanglement is a subtheory of $n$-partite entanglement in the analogous way for $k\le n$.
\end{example}

\begin{example}
The quantum resource theory of \em asymmetry \em with respect to a symmetry group $G$~\cite{gour2008resource,marvian2013theory} 
is defined in terms of the following partitioned process theory.  The systems are pairs $(\H,\pi)$ consisting of a Hilbert space $\H$ and a projective unitary representation $\pi$ of $G$ on $\mathcal{L}(\H)$, or equivalently, on the rays of $\H$. The processes $(\H,\pi)\lra (\H',\pi')$ are the completely positive trace-preserving maps from $\mathcal{\H}$ to $\mathcal{\H'}$, which can be sequentially composed in the usual way.
Parallel composition is given by the usual tensor product of Hilbert spaces and representations,
\[
(\H,\pi) \otimes (\H',\pi') := (\H\otimes\H',\pi\otimes\pi') .
\]
As in the theory of entanglement, the resulting SMC $(\C,\circ,\otimes,I)$ is equivalent to the usual SMC of quantum processes, since the representations $\pi$ are of no relevance to the morphisms and their composition.  These representations are needed, however, to define the distinguished subtheory of free processes.  
A process $(\H,\pi)\lra (\H',\pi')$ associated to the completely positive map $\Phi:  \mathcal{L}(\H)\to \mathcal{L}(\H')$ is free if it is covariant under the action of the group representation, i.e.,
\beq \label{eq:covariance}
\Phi \circ \pi(g) = \pi'(g) \circ \Phi \qquad\forall g\in G.
\eeq
In this case, $\Phi$ is said to be \em $G$-covariant\em.
It is straightforward to check that the group-covariant processes define a sub-SMC $\Cf$.

It follows in particular, that the free states on system $(\H,\pi)$
are the density operators that are invariant under the action of $G$, or \em $G$-invariant\em,
\begin{equation}\label{eq:Ginvariantstate}
\pi(g) \circ \rho = \rho \qquad \forall g\in G.
\end{equation}
The nonfree resources are then those states that are $G$-noninvariant, that is, those that fail to satisfy Eq.~\eqref{eq:Ginvariantstate}.

The quantum resource theory of asymmetry has many applications.  It is useful, for instance, when a quantum dynamical
problem cannot be solved exactly because it is too complex or because one lacks
precise knowledge of all of the relevant parameters, so that one must resort to
inferences based on a consideration of the symmetries. Measures of asymmetry are
perfectly adapted to making such inferences, for instance, they are constants of the
motion in closed-system dynamics. Indeed, in recent work, they have been shown to
provide a significant generalization of Noether's theorem~\cite{MarvianSpekkensNoether}. 
Another application includes the problem of achieving high-precision
measurement standards; for instance, the precision of a Cartesian frame is determined
by the extent to which it breaks rotational symmetry.

Note that one can just as easily define a resource theory of asymmetry for {\em classical} process theories or indeed for {\em any} process theory, including cases that are neither classical nor quantum such as the framework of generalized operational theories~\cite{ContPhys, Chiri2}. The key is that the categorical framework for process theories provides a straightforward means for defining a distinguished subtheory of {\em symmetric processes} as follows. Suppose the process theory is described by a category $\C$.  Then, as long as one can associate to every pair of objects $A,B \in |\C|$, representations $\pi$ and $\pi'$ of the group $G$, i.e. $\pi(g) \in C(A,A)$ and $\pi'(g) \in C(B,B)$, then any process $\Phi \in C(A,B)$ can be said to be $G$-covariant if it satisfies Eq.~\eqref{eq:covariance}.

\end{example}

\begin{example}
In the quantum resource theory of \em athermality \em with respect to a fixed temperature $T$~\cite{JWZGB,brandao2011resource} is defined in terms of the following partitioned process theory.  Systems are pairs $(\H,H)$ consisting of a Hilbert space $\H$ and a Hamiltonian $H$ acting on $\H$. Two such systems can be combined by taking tensor products of their Hilbert spaces and adding their Hamiltonians\footnote{In cases wherein two physical systems are interacting, they need to be described as a joint system which does not decompose into a parallel composition of its physical constituent systems. However, one can still try to find a \emph{different} tensor product decomposition with respect to which the Hamiltonian can be written as a sum of Hamiltonians on the tensor factors, i.e.~one can try to identify virtual subsystems which then decompose the joint system in the resource theory.},
\beq
\label{eq:productHam}
(\H,H)\otimes (\H',H') := (\H\otimes\H',H\otimes\mathbbm{1} + \mathbbm{1} \otimes H').
\eeq
Again, the processes are simply the completely positive trace-preserving maps, and the Hamiltonians are only relevant for the definition of the \em free \em processes. To wit, the sub-SMC of free processes is defined by a generating set that contains three kinds of free process: first, adding ancilla systems in Gibbs states with respect to temperature $T$, that is, states of the form $\frac{1}{Z}e^{-H/kT}$ where $Z={\rm tr}(e^{-H/kT})$ and $k$ is Boltzmann's constant; second, unitaries which are energy-preserving, i.e.~that commute with the total Hamiltonian; third, taking the partial trace over a subsystem.  These free processes are also known as \em $T$-thermal operations\em. It can be shown that a general process is $T$-thermal if and only if it has a Stinespring dilation whose ancilla state is the Gibbs state at temperature $T$ and whose unitary is energy-preserving, which means that it commutes with the Hamiltonian of system + ancilla.
Clearly, every system admits of only one free state, namely, the Gibbs state at temperature $T$ for that system.  Specifically, it is the state $\frac{1}{Z}e^{-H/kT}$ where $Z={\rm tr}(e^{-H/kT})$ and $H$ is the Hamiltonian for that system.  Any other state on that system is then a nonfree resource of athermality relative to the temperature $T$. These include states that are of the Gibbs form for a temperature $T'$ that is different from $T$, as well as states that are not of the Gibbs form at all.

As it turns out, the resource theory of nonuniformity is a full subcategory of the resource theory of $T$-athermality (for any $T$).  One simply needs to restrict the systems in the resource theory of $T$-athermality to those that have a trivial Hamiltonian (so that all states of the system have the same energy)~\cite{nonequilibrium}.  In this case, it suffices to specify the Hilbert space to specify the system; the constraint of being energy-preserving is trivial, so that all unitaries are allowed; and the Gibbs state (for any temperature) is just the completely mixed state.
\end{example}

\subsection{Resource theories of parallel-combinable processes}

Recall that states, which are processes of type $I\to A$, mapping nothing to something, are not generic. 
There are also general processes, which are morphisms of type $A\to B$ where $A$ is not necessarily isomorphic to the unit object $I$.  Such a process should be thought of as taking a nontrivial input.  
In a partitioned process theory, $(\C,\Cf)$, we can also 
consider the extent to which the resource states and resource transformations can be converted one to another by circuits composed entirely of free processes.  Resource theories that include transformations are typically richer than those including just states.  In some contexts, the distinction between resource states and general resource processes is described as {\em static} versus {\em dynamic} resources.

We will consider the resource theory of general processes.  In the next section, we will see that it is possible to make a distinction between combining resource processes in parallel and combining them in an arbitrary way.  In anticipation of this distinction, the resource theories we consider in this section will be termed resource theories of \em parallel-combinable \em processes.  The resource theory that is defined in terms of the process theory with distinguished subtheory $(\C,\Cf)$ will be denoted $\RC(\C,\Cf)$.

Suppose the process theory includes a process $f:A\to B$ and suppose that it is possible to embed $f$ into a circuit composed entirely of processes from the free set in such a way that the overall circuit implements the process $g:C\to D$, 
then we can say that we have successfully converted the resource process $f$ into the process $g$ using the free processes.
 
For example, in the theory of bipartite entanglement, this is precisely what happens with \em quantum teleportation\em:  by consuming the resource of a maximally entangled state on two qubits, one can simulate a single use of a quantum channel that transmits one qubit from Alice to Bob, using local operations and classical communication.  (We shall return to this example later.)

Given a partitioned process theory, $(\C, \Cf)$, we construct a new symmetric monoidal category $\RC(\C, \Cf)$ as follows. The objects of $\RC(\C,\Cf)$ are the processes of $(\C, \Cf)$, that is, $\bigcup_{A, B\in|\C|}\!\!\C(A, B)$.
When considering a particular process $f:A \to B$ as a resource, we think of it as a device which we have available in the laboratory and which we can compose with free operations, both in parallel and sequentially. Lemma~\ref{lm:combsall} will show that any such composition with free operations can be written as a circuit containing $f$ that has the form 
\beq\label{eq:comb}
\includegraphics{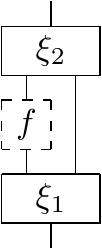}
\eeq
where $\xi_1:A' \to A \otimes Z$ and $\xi_2: B \otimes Z \to B'$ are free processes.
The circuit fragment that has a hole in place of the dashed box, and which takes as input any process of type $A \to B$ and transforms it into a process of type $A' \to B'$, is called a \em 1-comb, \em following the terminology of \cite{comb}.
In more precise terms, 
any 1-comb can be characterized by a triple $(Z, \xi_1, \xi_2)$ where $Z\in|\C|$, $\xi_1\in\Cf(A',A\otimes Z)$ and $\xi_2\in\Cf(B\otimes Z,B')$. One gets a process of type $A'\to B'$ by first applying the free operation $\xi_1$ to $A'$, then inserting a resource process $A\to B$ while doing nothing to $Z$, and finally applying the free operation $\xi_2$ to $B$ and $Z$.

\begin{lem}\label{lm:combsall}
Any circuit that contains a single occurrence of the process $f$ and only free processes otherwise can be put into the form of circuit~\eqref{eq:comb}, that is, in the form of a 1-comb that is built out of free processes, $\xi_1,\xi_2$, with the process $f$ slotted into the hole of the 1-comb.
\end{lem}

\begin{proof}
(Informal sketch) By a \em maximal parallel slice \em in a diagram we mean a list of objects $(Z_1, \ldots, Z_n)$ occurring in this diagram such that: (1) all pairs of objects included in it `occur in parallel', and (2) it is `not properly included in a larger parallel slice'. Given an expression in the language of SMCs, consider the diagram, and choose a maximal parallel slice which includes ${\rm dom}(f)$.  Rely on symmetry to achieve that the type of the slice is of the form ${\rm dom}(f)\otimes Z$, and we let $\xi_1$  be the morphism below the slice.  Now consider the same slice, with the exception that ${\rm dom}(f)$ is replaced by ${\rm codom}(f)$, and let $\xi_2$  be the morphism above the slice.  We have obtained the desired form.

\end{proof}

We say that two 1-combs are equivalent, $(Z,\xi_1,\xi_2)\sim(Z',\xi'_1,\xi'_2)$, if one has
\[
\includegraphics{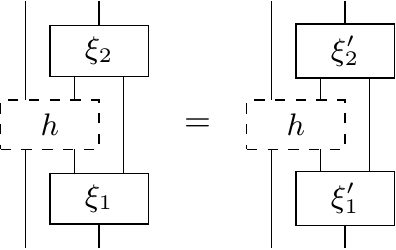}
\]
for any process $h$ between composite systems such that the types match. In other words, two 1-combs are equivalent as soon as they have the same operational behavior. The reason for introducing such an equivalence relation is that the auxiliary system $Z$ can be of arbitrary size; in particular, one can obtain a new $1$-comb from a given one by adjoining another auxiliary system to $Z$, which gets initialized in an arbitrary state and simply discarded afterwards. This new $1$-comb is equivalent to the original one, and as such it seems reasonable to regard both $1$-combs as implementing the same transformation between processes.

As the diagram shows, 1-combs can also be applied to processes between composite systems. However, as far as the possibility of transforming a process $f$ into a process $g$ is concerned, this does not yield any greater generality: the assumption that any identity morphism is a free process implies that $\xi_1$ and $\xi_2$ can be enlarged such as to comprise the additional input and output wires of $f$. 

In any case, we define the morphisms or transformations of type $f\to g$ in the category $\RC(\C,\Cf)$ to be the equivalence classes of 1-combs that turn $f$ into $g$. We use equivalence classes for the principal reason that taking equivalence classes should guarantee in all cases of interest that even if $\C$ is a large category, there is only a set (instead of a proper class) of equivalence classes of 1-combs that turn $f$ into $g$. For example in the quantum case, one can show this by proving that for any $f$ and $g$, there is a Hilbert space dimension $d$ such that any 1-comb $f\to g$ is equivalent to one in which the auxiliary system $Z$ has dimension $\leq d$. We regard such size issues as a minor concern that can safely be ignored, which is what we do in the following.

Sequential composition in $\RC(\C, \Cf)$ is as follows:
\[
\includegraphics{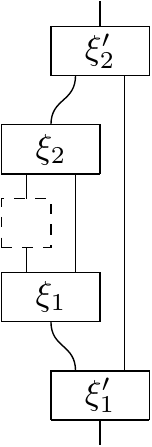}  
\]
This diagram forms a new 1-comb in the obvious way, and it is straightforward to check that this respects equivalence of 1-combs. Similarly, parallel composition works like this:
\[
\includegraphics{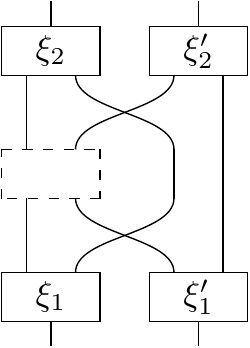}
\]
and the same remarks apply here. 

\begin{thm}
\label{thm:proctorec}
For any partitioned process theory $(\C,\Cf)$, the procedure outlined above allows one to define a symmetric monoidal category $\RC(\C,\Cf)$, and this can be interpreted as a resource theory in the sense of Definition~\ref{def:resourcetheory}.
\end{thm}

The proof is provided in Appendix~\ref{appendixA}.

Note that our construction of the category $\RC(\C,\Cf)$ is a souped-up version of the twisted arrow category construction~\cite[p.227]{MacLane}. More precisely, in the conventional twisted arrow category associated to $\C$, the objects are the processes $f$ and the morphisms between two such processes are the $1$-combs~\eqref{eq:comb} in which the ancilla system $Z$ is trivial. This conventional twisted arrow category can also be regarded as the category of elements of the hom-functor $\mathrm{hom} : \C \times \C^\mathrm{op} \to \Set$~\cite{nLabtac}; we suspect that our $\RC(\C,\C)$, or even the general $\RC(\C,\Cf)$, arises in a similar way from a variant of this construction for symmetric monoidal categories rather than plain categories.

We now present some examples of resource theories of parallel-combinable processes that can be defined in terms of a partitioned process theory.  We begin by demonstrating how the resource theory of communication channels of Example~\ref{ex:shannon} can be cast in this form.  

\begin{example}
\label{ex:Shannon}

The resource theory of \em two-way classical communication protocols\em.
We start with a variant of the SMC of classical stochastic processes of Example~\ref{ex:classicalstochasticprocesses}. The objects in this new category are pairs of finite sets, that is $(A,B)$ with $A,B\in|\FinStoch|$, whose parallel composition is defined componentwise,
\[
(A,B) \otimes (A',B') := (A\otimes A',B\otimes B').
\]
Processes of type $(A, B)\to (A', B')$ can then be taken to be the elements of $\FinStoch(A\otimes B, A'\otimes B')$, that is, as stochastic maps from $A \times B$ to $A' \times B'$, depicted as:
\[
\includegraphics{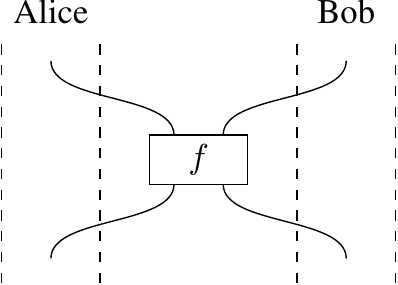}
\]
where the left and right regions refer to Alice and Bob respectively. These maps constitute a mathematical model for noisy two-way communication protocols.

The free processes are those that are generated by local operations and states of shared randomness. Just as in Example~\ref{ex:bientangle}, the local operations are defined to be those processes $(A,B)\to (A',B')$, represented by stochastic maps $\xi\in\FinStoch(A\otimes B,A'\otimes B')$, which can be factored into a parallel composition themselves,
\[
\xi = \xi_A \otimes \xi_B \: : \: A\otimes B \lra A'\otimes B',
\]
where $\xi_A \in\FinStoch(A , A')$ and $\xi_B \in\FinStoch(B,B')$.  These are depicted graphically as:
\[
\includegraphics{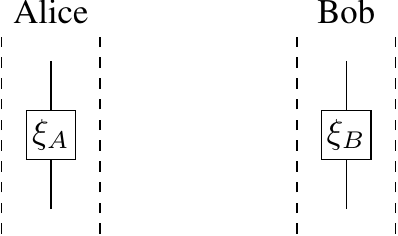}
\]
 The local operations include states $\xi:I \lra A\otimes B$ that factor into a product state $\xi_A \otimes \xi_B\: : \: I \lra A\otimes B$.  The free processes, however, include more than just the product states; they also include states of shared randomness, that is, states that are not of the product form.  In other words, \em all \em states  $\xi:I \lra A\otimes B$ are included in the free processes,
 \[
 \includegraphics{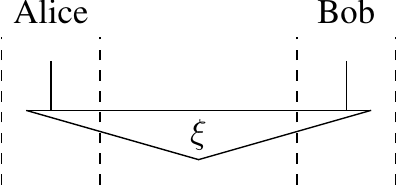}
\]
which models the assumption that shared randomness is free.

The nonfree resources, therefore, are the processes that have nontrivial inputs and which do not factor across the Alice-Bob partition.  These include processes that describe a single use of a communication channel from Alice to Bob, or from Bob to Alice.   
 For example, the process of Alice sending data to Bob (via a channel that has input type $A$ and output type $B$) has type $(A, I)\to (I, B)$ and can be depicted as a morphism $f\in\FinStoch(A\otimes I, I\otimes B)$ like this:
\[
\includegraphics{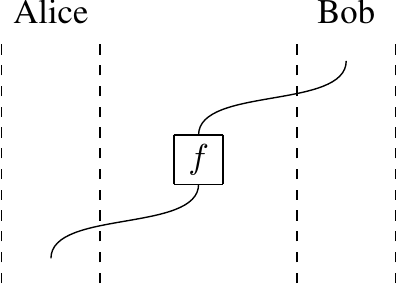}
\]

In the theory of communication channels, one is typically interested in knowing whether a channel $f:(A,I) \to (I,B)$ from Alice to Bob can simulate another channel $g:(A',I) \to (I,B')$.  To answer this question, one must consider the most general circuit of free processes that can be applied to $f$.  

In the case where we allow the sender and recipient to share randomness ahead of time, the processing of the channel looks like this:
\[
\includegraphics{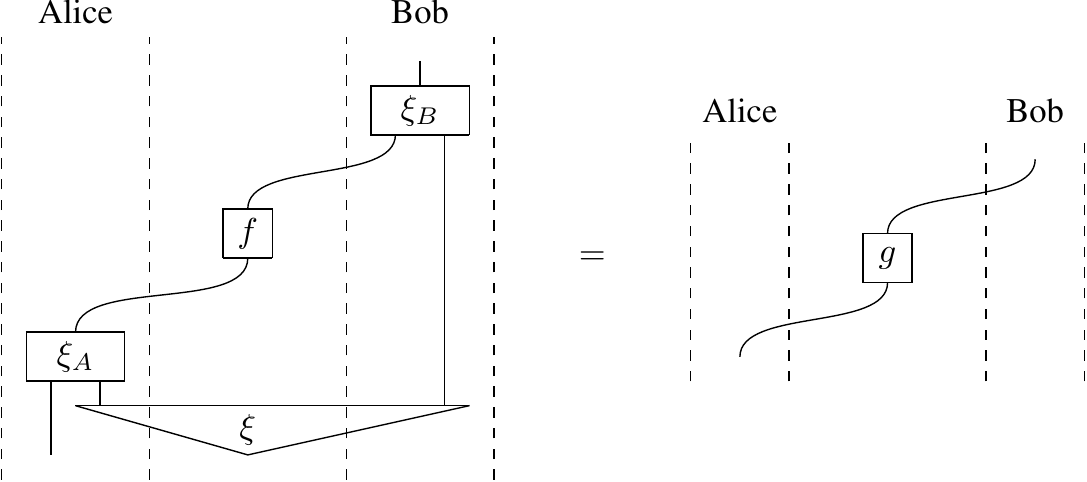}
\]
 One starts with a state of shared randomness $\xi: I \to S_A \otimes S_B$, then a free process $\xi_A: A'\otimes S_A\to A$ is applied to the input $A'$ and Alice's half of the shared randomness, $S_A$ to prepare the system $A$.  This is done in parallel with the identity process on Bob's half of the shared randomness, $S_B$.  Next, the channel $f:(A,I) \to (I,B)$ is implemented.  Finally, a free process $\xi_B: B \otimes S_B \to B'$ is implemented on $B$ and $S_B$.  
Note that here the system $S_B$ is an instance of the ancillary system $Z$ that featured in the general theory.  

The restriction of channels being only parallel-combinable arises if there can be no local processing between the uses of the channels.

The resource theory of two-way classical communication protocols contains as a proper subcategory the resource theory of one-way classical communication channels, which was considered in Example~\ref{ex:shannon}.  It suffices to consider the subcategory defined by the processes of type $(A,I)\mapsto(I,B)$, corresponding to channels from Alice to Bob, and the product states.

Note that the whole construction works likewise with any other SMC of processes in place of classical stochastic processes (see the next example, for instance) and can likewise be realised for any number of parties greater than two.
\end{example}

\begin{example}\label{ex:qcommchannels}
The resource theory of \em two-way quantum communication protocols\em.
We start with a variant of the SMC of quantum processes of Example~\ref{ex:quantumprocesses}.
The objects are pairs of Hilbert spaces, that is $(\H_A,\H_B)$, whose parallel composition is defined componentwise,
\[
(\H_A,\H_B) \otimes (\H_{A'},\H_{B'}) := (\H_A\otimes \H_{A'},\H_B\otimes \H_{B'}).
\]
The free processes include all the local processes, that is, the completely positive maps that factorize with respect to the tensor partition of the process theory, as well as \em all \em states  $s:I \lra A\otimes B$,  which models the fact that quantum and classical correlations are free in this resource theory (analogously to how shared randomness is free in the reource theory of two-way classical communication protocols).
  Everything in the resource theory of quantum communication can be represented graphically in a manner analogous to the theory of classical communication.
\end{example}

\begin{example}
In Example~\ref{ex:bientangle}, we considered the resource theory of {\em states} that resulted from considering the LOCC operations to be free processes.  One can also consider the more general resource theory of {\em processes} that LOCC defines.  Here, the resources include transformations and measurements in addition to states.  The nonfree resource processes include, in addition to the entangled states, the nonLOCC operations, which are sometimes called the ``entangling operations''. An example of such a nonfree resource process is a single use of a quantum channel. Unlike the resource theory of quantum communication channels of Example~\ref{ex:qcommchannels}, where {\em every} channel, classical or quantum, is a nonfree resource, in the theory of entanglement, the classical channels are included among the free processes. 
Here, the ancillary system $Z$ of the general theory is needed to accomodate the possibility of shared correlations between the two parties---either shared randomness or a shared entangled state---but also the possibility of classical communication.  

One can also consider the conversion of an entangled state into a quantum channel. The paradigmatic example of this is the quantum teleportation procotol, which is the conversion of a maximally entangled bipartite pure state into a single use of a noiseless quantum channel by LOCC.
For example, for two qubits,
by consuming the resource of a maximally entangled state $f$,
one can simulate a single use of a quantum channel $g$ that transmits one qubit from Alice to Bob, using local operations $\xi_{meas}$ (i.e.~a measurement) and $\xi_{contr}$ (i.e.~a controlled unitary) and classical communication $\xi_{comm}$.  
Graphically, this is depicted as follows:
\[
\includegraphics{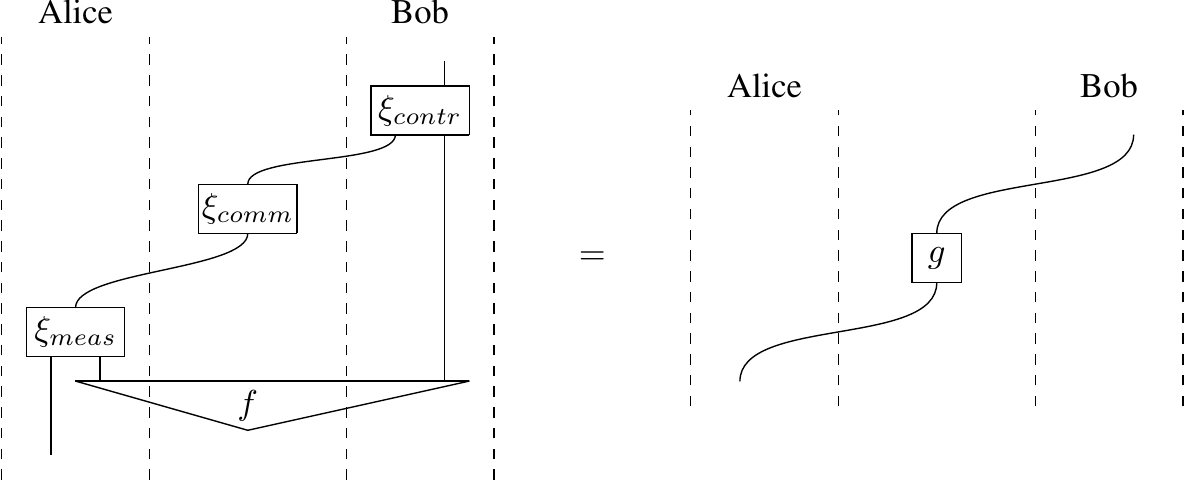}
\] 
\end{example}

\begin{example}
In asymmetry theory, the nonfree resource processes are the non-$G$-covariant operations.  For instance, in the resource theory of rotational asymmetry, the unitary that rotates a system about some spatial axis is a nonfree resource process, as is measuring the component of angular momentum of a system along some spatial axis~\cite{ahmadi2013wigner,marvian2012information}.
\end{example}

\begin{example}
In nonuniformity theory, the nonfree resource processes are the non-noisy operations.  An example is an erasure operation, taking an arbitrary state to a pure state.
\end{example}

\begin{example}
In athermality theory, the nonfree resource processes are the athermal operations. Example are heating, cooling, or doing work (for instance, lifting a weight).  An example of a conversion from a resource state to a general resource process is a heat engine.
\end{example}

\subsection{Resource theories of universally-combinable processes}
\label{sec:squentialresources}

The resource theory $\RC(\C,\Cf)$ has the following odd property. If one has access to two resource processes $f_1,f_2$, and one tries to use these in order to produce a certain target process $g$, then the transformation of type $f_1\otimes f_2\lra g$ looks like this,
\[
\includegraphics{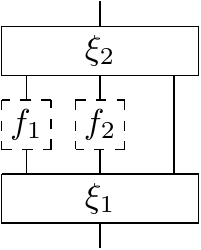}
\]
The problem with this kind of transformation is that $f_1$ and $f_2$ are necessarily composed in parallel. In particular, there is no way to try and produce $g$ by consuming $f_1$ and $f_2$ \em sequentially\em.

Here, we would like to define a different resource theory  in which there is no restriction at all on the manner in which a collection of resource processes may be consumed. Intuitively speaking, we regard resource processes as universally combinable, as in the laboratory where pieces of equipment can be put together in arbitrary ways. We term this the resource theory of universally-combinable processes and denote it by $\UC(\C,\Cf)$.

We now show how to construct $\UC(\C,\Cf)$ as a symmetric monoidal category in terms of $(\C,\Cf)$.
We take the objects of $\UC(\C,\Cf)$ to be finite sequences of processes $(f_1,\ldots,f_n)$. Such a sequence represents a collection of resource processes that we may have available. For all practical purposes, the ordering in the sequence is irrelevant; but in order to actually obtain an SMC, we still need to keep track of the ordering as a matter of syntax. The parallel composition of two objects $(f_1,\ldots,f_n)$ and $(g_1,\ldots,g_m)$ is given by concatenation of sequences,
\[
(f_1,\dots,f_n) \boxtimes (g_1,\ldots,g_m) := (f_1,\ldots,f_n,g_1,\ldots,g_m).
\]
We write ``$\boxtimes$'' instead of ``$\otimes$'' in order not to confuse the composition $\boxtimes$, which we think of as an abstract operation of combining collections of processes, with the parallel composition $\otimes$, which stands for actual parallel execution of processes. For a list $(f)$ containing only one process, we also omit the brackets and simply write $f$. Then $(f_1,\ldots,f_n)$ can be identified with $f_1\boxtimes\ldots\boxtimes f_n$.

Suppose that it is possible to embed each of the processes in the set $(f_1,\ldots,f_n)$ into a circuit composed entirely of processes from the free set in such a way that the overall circuit implements the process $g$, 
then we can say that we have successfully converted the collection of resource processes $f_1\boxtimes\ldots\boxtimes f_n$ into the process $g$ using the free processes.

Consider now a circuit that contains a single occurence of each of the processes $(f_1,\ldots,f_n)$.
In analogy with Lemma~\ref{lm:combsall}, it is not difficult to determine the most general way of converting a collection of processes $f_1\boxtimes\ldots\boxtimes f_n$ into a single process $g$.  One simply generalizes~\eqref{eq:comb}. The generic circuit that has holes in place of the dashed boxes, and which takes as input a collection of processes and transforms them into a single process, is called an \em $n$-comb, \em following the terminology of~\cite{comb}. It looks like this:

\beq\label{eq:ncomb} 
\includegraphics{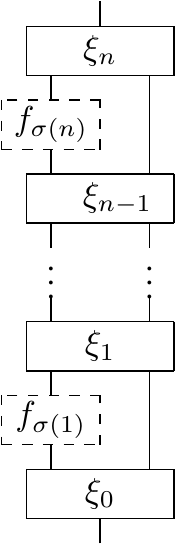}
\eeq
where $\sigma$ is a permutation on $\{1,\ldots,n\}$, and the $\xi_i$ are free processes. The additional wires parallel to the $f_{\sigma(i)}$ carry arbitrary systems $Z_i$. There is a straightforward definition of the equivalence of $n$-combs which generalizes equivalence of ordinary combs: two $n$-combs are equivalent if for all sequences of processes $(h_1,\ldots,h_n)$, where each $h_i$ may have composite input and output systems of which only part connects to the $n$-comb, the resulting composite process obtained by filling the blanks with the $h_{\sigma(i)}$ is the same.

One can now show that any method for combining the resource processes $f_i$ together with free processes into a composite process is of this form:

\begin{lem}
Any circuit that contains a single occurrence of each of the processes $(f_1,\ldots,f_n)$ and only free processes otherwise can be put in the form of the circuit~\eqref{eq:ncomb}, that is, in the form of an $n$-comb built of free processes, $\xi_0,\ldots, \xi_n$, with the processes $(f_1,\ldots,f_n)$ slotted into the holes of the comb in some order (given by the permutation $\sigma$).
\end{lem}

\begin{proof}
(Informal sketch)
We need to show that any expression in the language of symmetric monoidal categories that includes a single occurrence of each of the morphisms $(f_1,\ldots,f_n)$ and only free morphisms otherwise can be put in the form of such an $n$-comb.
We assume that the expression is given in terms of a diagram in the graphical calculus for SMCs, and choose a decomposition into layers as in~\cite[Sec.~1.3]{JS}. Then each layer contains at most one non-trivial process and otherwise wires that carry identity processes. If several consecutive layers contain only free processes, then these can be composed to a composite free process. In the layers containing the $f_i$, one can apply the symmetry in order to move any additional wires to the right of the resource process. This results in the desired form~\eqref{eq:ncomb}.
\end{proof}

At this level, the mathematical structure that most accurately captures these many-to-one transformations is that of a \em symmetric multicategory\em~\cite[2.2.21]{Leinster}, which we denote $\M(\C,\Cf)$; see also~\cite{lambek1989multicategories} for multicategories in general. The objects of this multicategory are precisely all the resource processes. The maps $(f_1,\ldots,f_n)\lra g$ in this multicategory are defined to be the equivalence classes of $n$-combs which transform $(f_1,\ldots,f_n)$ into $g$. These can be composed: given an $m$-comb $(g_1,\ldots,g_m)\lra h$ and an $n_j$-comb $(f_{j,1},\ldots,f_{j,n_j})\lra g_j$ for every $j=1,\ldots,m$, one obtains an $(n_1+\ldots+n_m)$-comb
\[
(f_{1,1},\ldots,f_{1,n_1},\ldots,f_{m,1},\ldots,f_{m,n_m}) \lra h
\]
by plugging the $n_j$-combs into the holes of the $m$-comb. There is also an obvious way to turn an $n$-comb $(f_1,\ldots,f_n)\lra g$ into a new $n$-comb $(f_{\sigma(1)},\ldots,f_{\sigma(n)})\lra g$ for any permutation $\sigma$. We leave it to the reader to verify that this data satisfies the axioms of a symmetric multicategory.

However, we also would like to be able to have transformations which output 
collections of processes, such as $g_1\boxtimes g_2$, rather than individual processes. What should such a transformation look like? One might be tempted to say that it should be an $n$-comb producing a composite process $g_1\otimes g_2$. While this might give rise to a resource theory as well, it is not the appropriate definition to make, since the constituents of a composite process like $g_1\otimes g_2$ cannot be accessed individually, but only in parallel. In particular, given a ``black box'' implementing the process $g_1\otimes g_2$, there is no way of obtaining the process $g_1\circ g_2$ (because one is constrained to a single use of the black box). Hence, according to the philosophy of ``universal combinability'' followed in this section, a transformation which produces $g_1\boxtimes g_2$ from a collection of resource processes $(f_1,\ldots,f_n)$ should produce $g_1$ from a subcollection of the $f_i$'s, and $g_2$ from the remaining $f_i$'s. In the laboratory picture, this corresponds to building $g_1$ from a collection of ingredients and building $g_2$ independently from another collection of ingredients.

The following is a general categorical construction for turning a symmetric multicategory into an SMC:
\begin{defn}
In the resource theory $\UC(\C,\Cf)$, a transformation of type $f_1\boxtimes\ldots\boxtimes f_n\lra g_1\boxtimes\ldots\boxtimes g_m$ consists of a function $\alpha:\set{1,\ldots,n} \to \set{1,\ldots,m}$ and maps
\[
(f_{k_1},\ldots,f_{k_{n_j}}) \lra g_j
\]
in the multicategory $\M(\C,\Cf)$, where $k_1,\ldots,k_{n_j}$ enumerates the elements of $\alpha^{-1}(j)$.
\end{defn}

Intuitively, the function $\alpha$ allocates the resource process $f_i$ to the production of $g_{\alpha(i)}$. The definition entails that every $f_i$ gets allocated to exactly one $g_j$ in this way; in particular, every $f_i$ has to be consumed somewhere. See Remark~\ref{rem:junk} for why we impose this.

It should be clear how to define parallel composition $\boxtimes$ of such transformations. Sequential composition is induced from the composition in the multicategory $\M(\C,\Cf)$. The details are tedious but straightforward, so we omit them.

\begin{example}
The resource theory of classical communication channels considered in Example~\ref{ex:Shannon} can be easily adapted to the case of universally-combinable processes.  This is the appropriate framework for the case wherein the uses of different channels can be interspersed with local processes.  An example of a conversion of channels $(f_{k_1},\ldots,f_{k_{n_j}}) \lra g_j$, for instance, is depicted as follows:
\[
\includegraphics{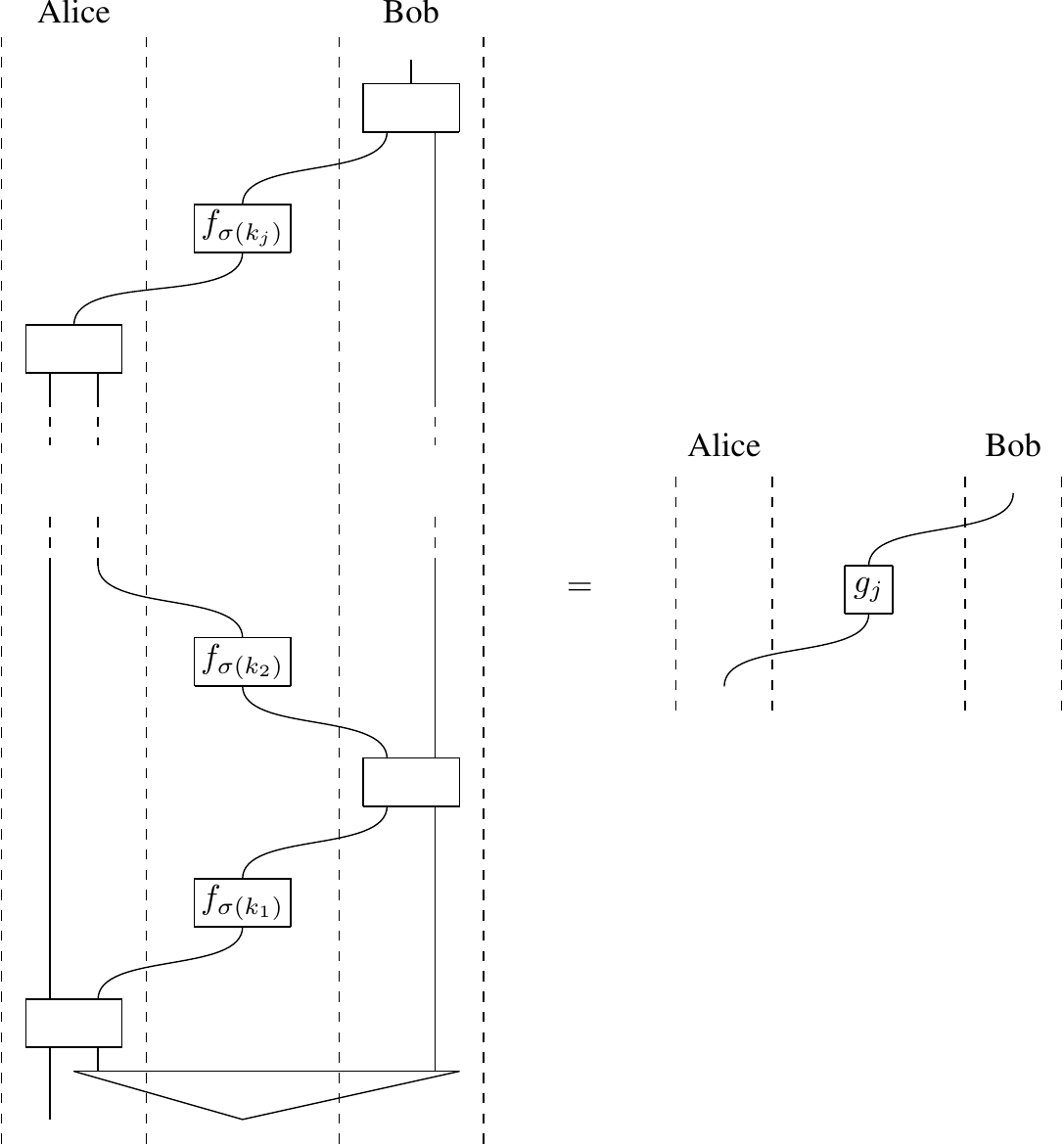}
\]

\end{example}

\section{Theories of resource convertibility}

\subsection{Definition}

Often, the most important questions that one asks of a resource theory $(\D,\circ,\otimes,I)$ concern whether or not a given resource conversion is possible, and not the particular process by which it occurs.
Given certain objects $A,B\in|\D|$, does there exist a transformation $A\to B$, i.e., is $\D(A,B)$ nonempty? 
Answering questions of this type does not require full knowledge of the category $\D$.
It is enough to know whether or not there exists a transformation $A\to B$; how many transformations there are of this type and how they can be found is also a relevant question, but one that we now would like to regard as secondary.
 With this in mind, we write $A\succeq B$ if there exists a transformation of type $A\to B$, and $A\not\succeq B$ if there does not. This ``$\succeq$'' is hence a preordering, meaning a reflexive and transitive binary relation 
\[
A\succeq A\qquad \mbox{and}\qquad A\succeq B,\: B\succeq C \:\Rightarrow A\succeq C.
\]
Intuitively, this says that every $A$ can be transformed into itself; and if $A$ can be transformed into $B$ and $B$ into $C$, then $A$ can also be transformed into $C$.

Now that we are no longer interested in the different morphisms of $\D$, we can as well use lowercase notation for the resource objects without the danger of confusion. This is what we do from now on in order to improve readability.

\begin{defn}
\label{defn:mereres}
A \em theory of resource convertibility \em $(R,+,\succeq,0)$ is a set $R$ equipped with a binary operation $+$, a preorder $\succeq$ and a distinguished element $0\in R$ such that for all $a, b, c \in R$,
\begin{align*}
a + (b + c) &\simeq (a + b) + c, & a + b & \simeq b + a, & a + 0 &\simeq 0 + a, 
\end{align*}
where $\simeq$ is the equivalence relation induced by the preorder, i.e.~$a\simeq b$ stands for $a\succeq b$ and $b \succeq a$.  We require these structures to interact as follows:   for all $a, b, c,d \in R$
\beq\label{eq:posetbefunct}
a \succeq b\ ,\ c \succeq d\  \Rightarrow\  a + c \succeq b + d .
\eeq
\end{defn}

In the context of enriched category theory~\cite{BorceuxStubbe}, this definition can be succinctly summarized by saying that a theory of resource convertibility is a symmetric monoidal category enriched over the poset of truth values $\{0,1\}$.  The correspondence to the definition in terms of the preorder is that $R(a,b)=1$ corresponds to $a\succeq b$, while $R(a,b)=0$ represents $a\not\succeq b$.  Also, the parallel composition ``$\otimes$'' of the SMC corresponds to the binary operation ``$+$'' of the theory of resource convertibility.

\begin{thm}
\label{thm:decat}
Let $(\D,\otimes,\circ,I)$ be a resource theory. If we define 
\[
R:= |\D|,
\]
and for all $a,b\in |\D|$,
\[
a + b := a\otimes b, 
\]
\[
a\succeq b \: :\Longleftrightarrow\: \D(a,b)\neq\emptyset,
\]
\[
0 := I,
\]
then $(R,+,\succeq,0)$ is a theory of resource convertibility.
\end{thm}

\begin{proof}
Straightforward.
\end{proof}

This construction captures the idea that we are only interested in whether a transformation of type $a\to b$ exists or not. One can think of it as a ``partial decategorification'': instead of \em completely \em forgetting the morphisms in $\D$ and the way in which they compose, the resulting theory of  resource convertibility only remembers a small remnant of categorical structure, namely whether there exists a morphism between any pair of objects or not.

In terms of enriched category theory, one can understand this construction as a change of base from enrichment over $\Set$ to enrichment over the poset of truth values $\{0,1\}$, along the monoidal functor $\Set\to\{0,1\}$ which maps the empty set to $0$ and any non-empty set to $1$.

One of the main goals when studying a resource theory is to find necessary and sufficient criteria for when a transformation of type $a\to b$ exists, i.e.~for when $a\succeq b$. In other words, one tries to find a way of characterizing the ordering relation $\succeq$. An answer to this problem typically consists in an algorithm which takes resource objects $a$ and $b$ as input and returns the answer to the question ``Is $a\succeq b$ or $a\not\succeq b$?'' 

The prototypical example is the theory of bipartite entanglement (Example~\ref{ex:bientangle}). One can start with the theory of bipartite entanglement as a partitioned process theory, defined by LOCC, as in Example~\ref{ex:bientangle}, one can turn it into a resource theory as in Theorem~\ref{thm:proctorec}, and then regard it as a theory of resource convertibility as in Theorem~\ref{thm:decat}.  Nielsen has shown that the convertibility relation between pure states can be algorithmically decided with the help of a majorization criterion, which constitutes a necessary and sufficient condition for $a\succeq b$~\cite{Nielsen}.

To emphasize the importance of the monoid structure, we will present two toy examples of resource theories that differ \em only \em in their monoid structure.

\begin{example}
\label{ex:food}
Consider the  resource theory of food, where for simplicity we only consider apples and bananas. A resource object is then a pair $(a,b)$ where $a,b\in\mathbb{N}$ stand for a number of apples and a number of bananas. E.g.~$(3,4)$ is the resource object ``3  apples and 4 bananas''. These can be combined via $\otimes$ in the obvious way by adding the numbers of each type of food,
\[
(a,b) \otimes (a',b') := (a+a',b+b').
\]
The only transformations that we declare to exist in this resource theory correspond to eating food: for example, there is one and only one morphism from $(2,3)$ to $(1,0)$, while there is no morphism going the other way. As is necessarily the case with a mathematical model of the real world, this toy model only captures a very small part of the phenomena of interest. A slightly more realistic model would include the gain of energy of a person from eating food as well as some agricultural processes used for producing food. The present example is as simplistic as possible in order to illustrate the mathematical structures involved.
\end{example} 

While this example has a strong quantitative aspect in the sense that the value of an object as a resource depends strongly on the number of pieces of each type of food, there are also resource theories which do not display this aspect at all:

\begin{example}
\label{ex:know}
While having two apples may be considered better than having only one apple, there also are entities for which possessing a large quantity is no better than possessing a small quantity. For example, \em proficiency \em has this property. For simplicity, let us assume that we only consider two sorts of proficiency, namely the proficiency level in arithmetic, measured by a number $a \in\mathbb{N}$, and the proficiency level in biology, measured by a number $b\in\mathbb{N}$. In other words, an object in our simplified resource theory of proficiency is a pair $(a,b)$ consisting of two proficiency levels $a,b\in\mathbb{N}$.

When combining two such pairs, for example when seeking to determine the proficiency level in each field of a group consisting of two people, one should take it to be the higher one of the constituent proficiency levels,
\[
(a,b) \otimes (a',b') := \left( \max(a,a'), \max(b,b') \right) .
\]
In other words, 
the tensor of two proficiency levels is the maximum of the two.
Having an expert in arithmetic and an expert in biology is just as good as having one person who is an expert in both, 
or
of having \em two \em people that are expert in both.

Similar to the case of the resource theory of food, the only kind of transformation that we consider is losing proficiency.
More precisely, we stipulate that there is one and only one transformation of type $(a,b)\to (a',b')$ if $a\geq a'$ and $b\geq b'$.
\end{example}

In both cases, a resource object $(a,b)$ can be transformed into $(a',b')$ if and only if $a\geq a'$ and $b\geq b'$. Hence the preorder ``$\succeq$'' can, for both theories, be illustrated by the Hasse diagram
\[
\includegraphics{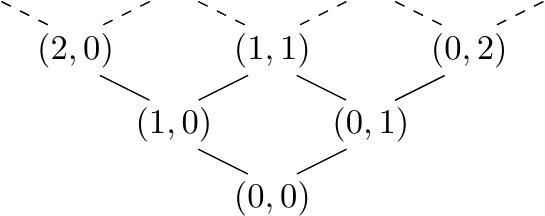}
\]
which is the partially ordered set
\[
(\mathbb{N}, \geq) \times(\mathbb{N}, \geq). 
\]
However, the binary operation ``$+$'' that defines the monoid structure is different: in the theory of food it is given by component-wise addition of natural numbers; in the theory of proficiency, it is given by the order-theoretic supremum. The theories have underlying commutative monoids
\[
(\mathbb{N}, +, 0) \times(\mathbb{N}, +, 0)\qquad{\rm vs.}\qquad (\mathbb{N}, \vee, 0) \times(\mathbb{N}, \vee, 0) .
\]
This difference makes the two theories behave quite differently. For example, catalysis (Definition~\ref{catalysis}) is impossible in the resource theory of food since borrowing food from someone else does not alleviate our hunger if we have to return the borrowed food eventually. In contrast, catalysis is possible in the resource theory of proficiency, since borrowing a proficient person who teaches us their knowledge does not degrade that person's proficiency.

\begin{rem}  
If $a$ and $b$ are objects in the original resource theory $\D$, then $a\simeq b$ does not mean that $a$ and $b$ are isomorphic in $\D$. Rather, it only means that there is a morphism of type $a\to b$ and one of type $b\to a$, but it does not follow that either composition $a\to b\to a$ or $b\to a\to b$ is necessarily equal to the identity morphism.

As long as we are only interested in \em whether \em a certain transformation $a\to b$ is possible, we should consider $\D$ as a theory of resource convertibility, and the adequate notion of $a$ and $b$ ``being the same'' is given by ``$\simeq$''. However, if we are also interested in \em how \em to turn $a$ into $b$, then we need to consider $\D$ as a full-fledged SMC, and the adequate notion of ``being the same'' is isomorphism. For example, if $a$ and $b$ are isomorphic, then there is a bijective correspondence between transformations $a\to c$ and transformations $b\to c$ for any $c$. However, if we only know that $a\simeq b$, then this is not necessarily the case; we only know that a transformation $a\to c$ \em exists \em if and only if a transformation $b\to c$ exists.
\end{rem}

\begin{rem}

Just as the appropriate notion of ``being the same'' between two categories is that of an equivalence of categories instead of isomorphism, the notion of being the same for two theories of resource convertibility, $R$ and $S$, is weaker than isomorphism. We say that $R$ and $S$ are \em equivalent\em, if there are functions $f:R\to S$ and $g:S\to R$ such that $g(f(a)) \simeq a$ for all $a\in R$, and $f(g(b)) \simeq b$ for all $b\in S$.

Alternatively, if we \em define \em two resource objects $a,b\in R$ to be equal, $a=b$, if $a\simeq b$, then this reduces precisely to the definition of isomorphism, which reads $g(f(a)) = a$ for all $a\in R$, and similarly the other way around. After all, $a\simeq b$ means that $a$ and $b$ are perfectly interchangeable as resource objects, and as such there is no need to distinguish them. For this reason, we believe that ``$\simeq$'' is the appropriate definition of equality of resource objects. So although we always write ``$\simeq$'' in the context of a theory of resource convertibility, there is no harm or loss of generality in replacing this by ``$=$''.
\end{rem}

\begin{remark}
There are examples of mathematical structures satisfying the conditions of Definition~\ref{defn:mereres} that do not have an obvious interpretation in terms of resource theories. One such example is lattice ordered groups~\cite{BirkhoffLatticeGroup}, and another one is the commutative case of quantales \cite{Mulvey},  where there is a distributive law between the monoid multiplication and the suprema of the ordering. Yet another example is the Cuntz semigroup from $C^*$-algebra theory~\cite{Cuntz}. We hope that this makes it clear that Definition~\ref{defn:mereres} is also of purely mathematical interest, and illustrates part of the wide range of phenomena that it captures.
\end{remark}

\subsection{A bit of phenomenology}

We now describe some of the phenomena that may occur in theories of resource convertibility as in Definition~\ref{defn:mereres} and prove some general results about when these phenomena occur and when they do not. We regard this as the very beginnings of a comprehensive study of the phenomenology of theories of resource convertibility and the development of general results about them. Throughout, $R$ is a theory of resource convertibility.

The first such phenomenon that we would like to discuss has its origin in the resource theory of chemistry (Example~\ref{ex:chem}): \em catalysis\em.

\begin{defn}
\label{catalysis}
A resource $c\in R$ is a \em catalyst \em for $a,b\in R$ if $a\not\succeq b$, but $a+c \succeq b+c$.
\end{defn}

So although there is no way to transform $a$ into $b$ in $R$, it \em is \em possible to do so if one also has the catalyst $c$ available, since then one can turn $a+c$ into $b+c$. The power of a catalyst lies in the fact that $c$ is also a result of the transformation, and hence can be reused in other transformations. Only one copy of the catalyst is needed in order to transform any number of $a$'s into $b$'s.

\begin{example}
\label{ex:catalyst}
In the resource theory of compass-and-straightedge constructions, presented in the introduction,
 the figure consisting of a circle of unit radius and a square of area $\pi$ can be regarded as a catalyst for the transformation of turning a circle into a square of the same area, i.e.~for ``squaring the circle''. This works as follows: for a given circle, start by constructing its center. Its radius is hence available as a line segment, which can be compared with the radius of the unit circle. Scaling the reference square of area $\pi$ by the same factor yields a square of the same area as the original circle.
\end{example}

It is an important question to determine whether the transformation of some resource $a$ into some resource $b$ is possible with the help of a catalyst, or if it remains impossible even with any other resource as a candidate catalyst. More specifically, it is also relevant to know whether catalysis is possible at all in $R$:

\begin{defn}
$R$ is \em catalysis-free \em if 
\[
 a + c \succeq b + c \:\Longrightarrow\: a\succeq b.
\]
\end{defn}

In order to showcase the possibility of deducing general mathematical theorems that apply to all theories of resource convertibility, we now derive a criterion for when such a theory is catalysis-free. This involves some other properties that a theory of resource convertibility may or may not possess and which we now define.

\begin{defn}
\label{def:nonintact}
$R$ is \em non-interacting \em if
\[
a\succeq b_1 + b_2 \quad\Longrightarrow\quad \exists a_1,a_2\in R, \: a \simeq a_1 + a_2,\: a_1 \succeq b_1, \: a_2\succeq b_2.
\]
\end{defn}

Intuitively, $R$ is non-interacting if every transformation which outputs a combination of two resources can be decomposed into two transformations, each of which outputs a constituent resource. In the theory of ordered abelian groups, this condition is also known as the \em Riesz decomposition property\em.

As a general class of examples, take the resource theory $\UC(\C,\Cf)$ generated by a theory with free processes $(\C,\Cf)$.  Such a theory is non-interacting by definition, since in order to produce a collection of resource processes, one needs to produce each one individually.

\begin{defn}
$R$ is \em quantity-like \em if 
\[
a_1 + a_2 \simeq b_1 + b_2, \: a_1\succeq b_1 \quad\Longrightarrow\quad b_2\succeq a_2
\]
\end{defn}

One may think of this as follows: if $a_1+a_2\simeq b_1+b_2$, then $a_1+a_2$ and $b_1+b_2$ are of equal value as resources; but by the assumption, $a_1\succeq b_1$, the resource $a_1$ is at least as valuable as $b_1$. Hence, assuming a certain conservation-of-value property, we conclude that $b_2$ must be at least as valuable as $a_2$. Since we have not made precise the notion of value, this explanation is nothing but an intuition. In a quantity-like resource theory, the resource objects typically have an extensive flavour, like ``volume'' or ``mass''.

For example, the resource theory of food (Example~\ref{ex:food}) is quantity-like, while the theory of proficiency (Example~\ref{ex:know}) is not.

\begin{thm}
If $R$ is non-interacting and quantity-like, then $R$ is catalysis-free.
\end{thm}

\begin{proof}
Assume that $a + c \succeq b + c$. Since $R$ is non-interacting, we need to have $a + c \succeq a_1 + a_2$ with $a_1 \succeq b$ and $a_2 \succeq c$. Because $R$ is quantity-like, $a + c \succeq a_1 + a_2$ and $a_2 \succeq c$ implies $a \succeq a_1$. From $a \succeq a_1$ and $a_1 \succeq b$ and using transitivity, we conclude that $a\succeq b$.
\end{proof}

This is our first general result about theories of resource convertibility. It provides a non-trivial sufficient criterion for the absence of catalysis.

\begin{prop}[No-cloning]
\label{prop:noclone}
If $R$ is quantity-like, then
\[
a \succeq a+a \quad\Longleftrightarrow\quad 0 \succeq a.
\]
\end{prop}

This result states that for a theory of resource convertibility that is quantity-like, a resource $a$ can be cloned if and only if it is actually worthless, i.e.~it can be produced from nothing. In other words, no non-trivial resource can be cloned. In particular, this means that each non-trivial resource is finite in character: if it was infinite, like the number of rooms available in Hilbert's hotel, then one could produce clones of it without any cost.

\begin{proof}
We have $a\succeq a$, hence $0\succeq a$ trivially implies that $a+0\succeq a+a$. Conversely, since $R$ is quantity-like, $a+0\succeq a+a$ and $a\succeq a$ imply that $0\succeq a$.
\end{proof}

\begin{example}
\label{ex:math}
If we consider the resource theory wherein mathematical propositions are the objects and proofs are the morphism (described in the introduction) with \em linear logic \em \cite{Girard} in the background, then 
resource objects also cannot be cloned. For this reason, linear logic is often referred to as a ``resource sensitive logic'', or in Girard's own words:  ``While classical logic is about truth, linear logic is about food''.
\end{example}

The opposite extreme to being quantity-like is this:

\begin{defn}
$R$ is \em quality-like \em if $a+a\simeq a$ for all $a\in R$.
\end{defn}

For example, the resource theory of proficiency (Example~\ref{ex:know}) is quality-like, while the resource theory of food (Example~\ref{ex:food}) is not. Being quality-like means that the number of copies that one has available of a resource is irrelevant.

\begin{prop}
If $R$ is quality-like, then the following conditions are equivalent for any $a,b\in R$:
\begin{enumerate}
\item $a+a \succeq b$,
\item $a\succeq b$,
\item $a \succeq b+b$.
\end{enumerate}
Conversely, if these conditions are equivalent, then $R$ is quality-like.
\end{prop}

The proof is straightforward.

In Proposition~\ref{prop:noclone}, we have encountered the condition $0\succeq a$, i.e.~that $a$ can be produced from nothing. Dually, we can ask whether $a$ can be turned into nothing:

\begin{defn}
A resource $a\in R$ is \em freely disposable \em if $a\succeq 0$.
\end{defn}

\begin{rem}
\label{rem:junk}
In many real-world examples, disposing a resource does itself come at a cost, and the given resource can be thought of as having a negative value; think e.g.~of nuclear waste, the disposal of which requires a sizeable investment of other resources, such as treatment, safe storage, and decay time. In order to capture this sort of phenomenon, the general theory of resource theories should not assume that all resources are freely disposable.
\end{rem}

\begin{defn}
If every $a\in R$ is freely disposable, then we say that $R$ is \em waste-free\em.
\end{defn}

One should keep in mind that this is a mathematical definition which may be satisfied for a given resource theory even when its resource objects can not be safely disposed of in practice. 

Being waste-free is equivalent to $0\in R$ being a bottom element of the preorder $\succeq$. In other words, it means that if $0\succeq a$, then $a\simeq 0$. If this is the case, then it follows that for all $a,b\in R$,
\beq
\label{eq:nonneg}
a + b \succeq a,
\eeq
which means that the operation $\blank + b: R\to R$ is order-nondecreasing for any $b$. We may think of this property as saying that, as a resource, the whole is always greater than or equal to each of its parts. If $R$ is not waste-free, then this is not necessarily the case.

\begin{example}
If $(\C,\Cf)$ is a theory with free processes such that $\C(I,I)=\{1_I\}$ and such that for every object $A\in\C$, there exists a process $\psi_A:I\to A$ and a process $\phi_A:A\to I$, then both resource theories of processes $\RC(\C,\Cf)$ and $\UC(\C,\Cf)$ are waste-free. This is easy to see by plugging the open ends of an undesired process $f:A\to B$,
\[
\includegraphics{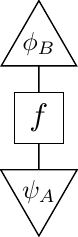}
\]
Since this is a process of type $I\to I$, it must necessarily be $1_I$, which is exactly the trivial resource $0$.
\end{example}

Finally, we end this section with one last property that a theory of resource convertibility may or may not have and that may be of interest:

\begin{defn}
$R$ has the \em Riesz interpolation property \em  if $(a_i)_{i\in I}$ and $(b_j)_{j\in J}$ are finite families of resources with $a_i\succeq b_j$ for all $i$ and $j$, then there exists $c\in R$ such that
\[
a_i\succeq c\succeq b_j.
\]
\end{defn}

This is a well-known property for partially ordered sets. In our context of resource theories, we think of $c$ as a ``proxy resource'' for turning an $a_i$ into a $b_j$.

\section{Quantitative concepts for theories of resource convertibility} 

\subsection{Monotones}

There is a long tradition in science of ``measuring'' or ``quantifying'' the utility of a resource object in terms of a number assigned to it. In our approach, this takes the following form:

\begin{defn}
A \em monotone \em on $R$ is a function $M:R\to\mathbb{R}$ such that
\[
a\succeq b \quad\Longrightarrow\quad M(a) \geq M(b).
\]
\end{defn}

The main use of a monotone lies in detecting the non-convertibility of a resource $a$ into a resource $b$: if $M(a) < M(b)$, then $a\not\succeq b$. The converse is in general not true, since $\succeq$ is only a preorder, and therefore may not be totally ordered. In this case, it is not possible to capture the ``value'' of a resource in terms of a single number. However, a \em family \em of monotones $(M_i)_{i\in I}$ may be \em complete \em in the sense that the monotones in this family characterize the order completely,
\[
a\succeq b \quad\Longleftrightarrow\quad M_i(a)\geq M_i(b) \:\forall i\in I.
\]

\begin{prop}
For any $R$ there exists a complete family of monotones.
\end{prop}

The proof is very simple.

\begin{proof}
We take the index set of the family to be $I:=R$ itself. Then for any $i\in I$, define the monotone $M_i:R\to\mathbb{R}$ via
\[
M_i(a) := \begin{cases} 1 & \textrm{if } a\succeq i\\
                        0 & \textrm{if } a\not\succeq i. \end{cases}
\]
This is a monotone thanks to transitivity of $\succeq$. To show completeness, we start with a pair $a,b\in R$ with $a\not\succeq b$ and find a monotone in the family which witnesses this. And indeed, $M_b$ does the job, since $M_b(b)=1$, but $M_b(a)=0$. 
\end{proof}

Intuitively, we have simply defined, for every resource, the monotone that describes whether or not any other resource is convertible to it or not.  Given this monotone for every resource, one can obviously deduce whether any given conversion is possible or not, hence the set of monotones is complete. In terms of enriched category theory, one can understand the monotone $M_i$ to be a hom-functor, and the proposition is a special case of the enriched Yoneda lemma.

What this result demonstrates is that the only interesting question is whether one can find a complete set of monotones with a number of elements less than the number of elements in $R$, or in the case where there are infinitely many resources, whether one can find a complete set of monotones that can be parameterized with fewer parameters than is required to parameterize $R$.

One may think of a monotone as the assignment of a price to each resource in a way which is consistent with the convertibility relation. In this sense, one can also think of the letter ``$M$'' as standing for a consistent \em market \em in which the resource objects of $R$ are being traded.

The definition implies immediately that if $a\simeq b$, then $M(a)=M(b)$.

We can distinguish certain kinds of monotones having special properties, 

\begin{defn}
A monotone $M$ is
\begin{enumerate}
\item 
\em additive \em if
\[
M(a+b) = M(a) + M(b).
\]
\item \em supremal \em if
\[
M(a+b) = \max\{M(a),M(b)\}.
\]
\end{enumerate}
\end{defn}

The first condition implies in particular that $M(0)=0$. For the second condition, one may also assume $M(0)=0$ without loss of generality: defining $M'(a):=M(a)-M(0)$ yields a new supremal monotone satisfying $M'(0)=0$ while witnessing the same impossible conversions $a\not\succeq b$ as the original $M$.

If one regards $\mathbb{R}$ itself as a theory of resource convertibility, where the ordering is the usual one and addition is also the usual one, then an additive monotone is simply a homomorphism of theories of resource convertibility $M:R\to\mathbb{R}$. If one takes $\max:\mathbb{R}\times\mathbb{R}\to\mathbb{R}$ as the binary operation instead, then the definition of a supremal monotone (with $M(0)=0$) can be neatly summarised by saying that it should be a homomorphism of theories of resource convertibility $M:R\to\mathbb{R}$.

We suspect that additive monotones are most useful in resource theories that have a quantitative flavour, and in particular whenever $R$ is quantity-like. In contrast, a supremal monotone should behave more like a measure of quality and should be especially relevant when $R$ is quality-like.

\begin{example}
In the resource theory of food (Example~\ref{ex:food}), an additive monotone is determined by assigning a price $M((1,0))$ to an apple and a price $M((0,1))$ to a banana. Since $(1,0)\succeq 0$ and $(0,1)\succeq 0$, both of these prices have to be non-negative. By additivity, we necessarily must have
\[
M((a,b)) = a\cdot M((1,0)) + b\cdot M((0,1)),
\]
so that $M$ simply determines the total price of a resource object by tagging each apple and each banana with its price and adding the numbers up. Similarly, a supremal monotone is also determined by the price for an apple and the price for a banana; the price assigned to a collection of food given by a resource object $(a,b)$ is then given by the highest price of any of its constituents. 
\end{example}

\subsection{Conversion rates}

If $a\succeq b$ in a theory of resource convertibility, then certainly $a+a\succeq b+b$ as well. Another important phenomenon that may occur in resource theories is that the converse is not necessarily true: it may be the case that $a+a\succeq b+b$, although $a\not\succeq b$. In other words, it may be the case that one may be able to produce two copies from $b$ from two copies of $a$, although it is not possible to transform $a$ into $b$ at the single-copy level. 
 It is one of several ``economy of scale'' effects which makes mass production more efficient than small-scale production. See~\cite{BRS} for concrete examples of this activation phenomenon in the resource theory of bipartite entanglement (Example~\ref{ex:bientangle}). It is not limited to comparing the one-copy level with the two-copy level; we can compare any two distinct numbers of copies.  Introducing the notation $n\cdot a:= \underbrace{a + \ldots + a}_n$ for the resource corresponding to $n$ copies of $a$, it may happen that $k\cdot a\not\succeq k\cdot b$, although $n\cdot a\succeq n\cdot b$ for some suitable $n>k$.
However, in resource theories that are quality-like, this kind of behaviour is obviously not possible.

In general, the following is considered a particularly important question in a resource theory: in a setting in which one can turn many copies of $a$ into many copies of $b$, how many copies of $a$ does one need on average in order to produce one copy of $b$? In effect, this question is about the \em rate \em of conversion of $a$ to $b$:

\begin{defn}\label{def:confrate}\em
For two resources $a,b\in R$ of a theory of resource convertibility, the \em conversion rate \em from $a$ to $b$, or simply \em rate \em from $a$ to $b$, is:
\beq
\label{eq:defnrate}
R(a\to b):= 
 \sup\left\{  \frac{m}{n}\Bigm| n\cdot a \succeq m\cdot b,\:  n, m \in\mathbb{N}_+\right\} .
\eeq
\end{defn} 

If $n\cdot a\not\succeq m\cdot b$ for all positive integers $n$ and $m$, then we define $R(a\to b):=0$; in all other cases, the rate is a strictly positive real number or $\infty$. This definition concerns the \em maximal \em rate; if we take the infimum in~\eqref{eq:defnrate} as opposed to the supremum, we obtain the analogous definition of \em minimal \em rate. Which one of these two one would like to attain depends significantly on the context: if $b$ is a resource object that one desires to produce, the maximal rate~\eqref{eq:defnrate} is the relevant quantity; if $a$ is a ``negative resource'' that one would like to get rid of (Remark~\ref{rem:junk}) by converting it into copies of a less problematic object $b$, then the minimal rate will be the appropriate figure to look at.

Extensive monotones yield upper bounds on the maximal rate:

\begin{thm}\label{thm:rate}
Let $M$ be an extensive monotone. Then
\beq\label{eq:ratemu}
R(a\to b) \cdot M(b) \leq M(a).
\eeq
\end{thm}

It is straightforward to derive the same result for the minimal rate, but with the inequality sign going the other way.

\begin{proof}
If $n\cdot a\succeq m\cdot b$, then we obtain by extensivity of $M$,
\[
n \cdot M(a) \geq m\cdot M(b).
\]
Rearranging to
\[
\frac{m}{n} \cdot M(b) \leq M(a)
\]
proves the claim upon taking the supremum corresponding to~\eqref{eq:defnrate}.
\end{proof}

If $M(b)>0$, then one can rearrange~\eqref{eq:ratemu} to the more natural-looking inequality
\[
R(a\to b)\leq \frac{M(a)}{M(b)}.
\]
If one thinks of $M$ as a consistent market which assigns to each resource object its price, then this makes good sense: the number of $b$'s that one can obtain on average from one unit of $a$ cannot be higher than the value of one unit of $a$ relative to one unit of $b$.

\section{Closing}\label{closing}

We have defined resource theories as symmetric monoidal categories, the objects of which are resources which the morphisms transform into each other. We have explained several ways in which one can obtain a resource theory from a category of processes equipped with a distinguished subcategory of ``free'' processes. Finally, we have noted that the questions that one typically asks about a resource theory---concerning catalysis, rates of conversion, etc.---only concern the question of \em whether \em a transformation of a certain type exists or not. This led to the definition of a theory of resource convertibility, in which the structure of a category is ``decategorified'' to the structure of a commutative preordered monoid. Our first small results are nothing but the very beginnings of a detailed investigation of theories of resource convertibility and their mathematical structure. Our long list of examples suggests that there will be a strong interplay between the abstract and general theory and the phenomenology of concrete resource theories of interest.

A major problem that we have not yet touched upon is \em epsilonification\em. The idea here is that in many applications, such as in the resource theory of communication (Example~\ref{ex:shannon}), it may be sufficient to turn a given resource $a$ into some $b'$ \em close to \em a target resource $b$, i.e.~turning $a$ into $b$ ``up to $\varepsilon$''. Typically, one would like the $\varepsilon$ to become arbitrarily small, possibly as the number of copies of $a$ and $b$ increases. We are currently investigating how our definitions should be modified in order to be able to deal with this kind of question as well. One way to do this to define an \emph{epsilonified theory of resource convertibility} just as in Definition~\ref{defn:mereres}, but replacing the preorder by a \emph{cost function}, which is a function assigning to every two $a,b\in R$ a number $c(a,b)\in \mathbb{R}_{\geq 0}$ which measures how close one to get to $b$ by consuming $a$, and such that suitable axioms hold which are analogous to the ones of Definition~\ref{defn:mereres}.

\appendix

\section{Proof of Theorem~\ref{thm:proctorec}}\label{appendixA}

\begin{proof}
What we really need to show is that $\RC(\C,\Cf)$ is an SMC.

Associativity of sequential  and parallel composition are obvious in the diagrammatic notation. The identity on $f$ is the comb $(I, 1_{{\rm dom}(f)}, 1_{{\rm codom}(f)})$, and the tensor unit is $1_I$.  

Using bifunctoriality of the tensor product in $\C$, we show bifunctoriality of the tensor product in $\RC(\C,\Cf)$:
\[
\includegraphics{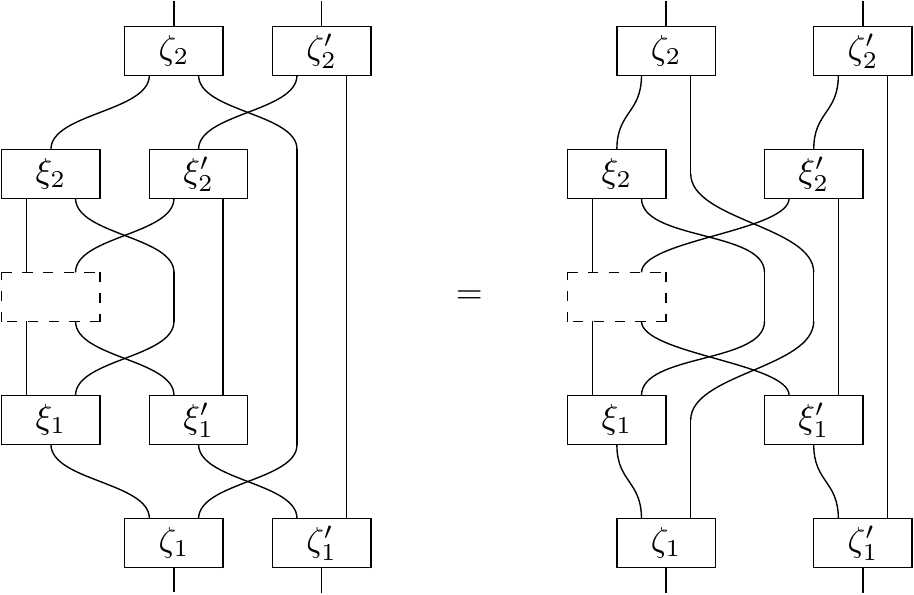}
\]
Here, the left-hand side defines the comb obtained by composing first in parallel and then sequentially,
\[
\left((Z_\zeta, \zeta_1, \zeta_2)\otimes (Z'_\zeta, \zeta_1', \zeta_2')\right)\circ\left((Z_\xi, \xi_1, \xi_2)\otimes (Z'_\xi, \xi_1', \xi_2')\right)\ ,
\]
while the right-hand is the comb coming from composing first sequentially and then in parallel,
\[
\left((Z_\zeta, \zeta_1, \zeta_2)\circ (Z_\xi, \xi_1, \xi_2)\right)\otimes\left((Z'_\zeta, \zeta_1', \zeta_2')\circ (Z'_\xi, \xi_1', \xi_2')\right)\ .
\]
The equality of the two composite combs follows from the invariance properties of the graphical calculus for SMCs.

Finally, the symmetry natural isomorphism on the object $f\otimes g$ is the comb
\[
\includegraphics{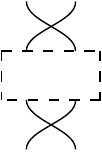}
\]
which corresponds to the triple $(I,\sigma_1,\sigma_2)$, where
\[
\sigma_1:{\rm dom}(g)\otimes {\rm dom}(f) \lra {\rm dom}(f) \otimes {\rm dom}(g),\qquad \sigma_2:{\rm codom}(f)\otimes {\rm codom}(g) \lra {\rm codom}(g) \otimes {\rm codom}(f)
\]
are the symmetry isomorphisms. These are indeed free processes because $\Cf$ is a sub-SMC of $\C$. Naturality of the symmetry follows from the equality
\[
\includegraphics{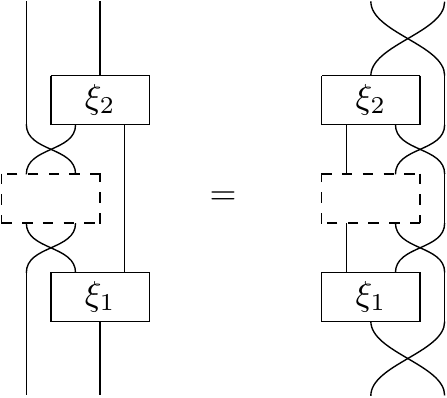}
\]
Here, the left-hand side depicts an application of the symmetry before the comb transformation, while the right-hand side applies the comb transformation before the symmetry. Equality of these two diagrams follows again from the graphical calculus.
\end{proof}

\bibliographystyle{plain}
\bibliography{Resource}

\end{document}